\documentclass[reqno,11pt]{article} 

\usepackage{amsfonts,amssymb,amsmath,amsthm,mathrsfs}
\usepackage{color}
\usepackage{tikz} 
\usepackage{hyperref} 
\usepackage{caption}

\usepackage{marvosym}

\makeatletter 
\newcommand\ackname{Acknowledgements}
\if@titlepage
\newenvironment{acknowledgements}{
   \titlepage
   \null\vfil
   \@beginparpenalty\@lowpenalty
   \begin{center}
      \bfseries \ackname
      \@endparpenalty\@M
   \end{center}}
{\par\vfil\null\endtitlepage}
\else

\fi
\makeatother

\hypersetup{
   pdfmenubar=false,        
   pdfauthor={Eris Runa},     
   pdfsubject={Subject},   
   pdfcreator={Eris Runa},   
   pdfproducer={Eris Runa}, 
   pdfnewwindow=true,      
   colorlinks=false,       
   linkcolor=blue,          
   citecolor=blue,        
   filecolor=magenta,      
   urlcolor=cyan           
}

\usepackage[left=1in,right=1in,top=1.3in,bottom=1.4in]{geometry}

\newcommand{\homo}{{ \textrm{hom} }}

\usepackage{xspace}
\usepackage{enumitem}

\newcommand{\Z}{\mathbb{Z}}
\newcommand{\sbv}{\textrm{SBV}}
\newcommand{\sbvp}{\textrm{SBV}_p}
\def\M{\mathbb M}
\newcommand{\bv}{\textrm{BV}}

\newcommand{\leb}{\mathcal{L}}

\def\ae{\mbox{\,a.e.\xspace}}

\newcommand{\dx}{{\,\mathrm{d}x}}

\newcommand{\dt}{{\,\mathrm{d}t}}

\newcommand{\dy}{{\,\mathrm{d}y}}

\newcommand{\N}{\mathbb{N}}
\newcommand{\R}{\mathbb{R}}
\newcommand{\Q}{\mathbb{Q}}

\newcommand\scal[2]{{\left\langle #1 ,#2\right\rangle}}

\newcommand{\ottoboh}[1]{ }
\newcommand{\akmboh}[1]{{\color{green}}}
\newcommand{\BrascampLiebBoh}[1]{}
\newcommand{\res}{\mathop{\hbox{\vrule height 7pt width .5pt depth 0pt \vrule height .5pt width 6pt depth 0pt}}\nolimits}

\newcommand{\hausd}{\mathcal H}

\newcommand{\eps}{\varepsilon}
\newcommand{\ro}{\varrho}

\def\d{\,\mathrm{d}}

\newcommand{\insieme}[1]{\left \{#1\right \}}

\newcommand{\ie}{i.e.,\xspace}

\DeclareMathOperator{\dist}{dist}

\DeclareMathOperator{\diam}{ diam}

\newtheorem{theorem}{Theorem}[section]

\newtheorem{lemma}[theorem]{Lemma}

\newtheorem{remark}[theorem]{Remark}
\newtheorem{rmk}[theorem]{Remark}

\newtheorem{proposition}[theorem]{Proposition}
\newtheorem{prop}[theorem]{Proposition}

\newtheorem{corollary}[theorem]{Corollary}

\title{Sobolev and SBV Representation Theorems for large volume limit Gibbs measures.} 

\usepackage{authblk}
\author{Eris Runa\thanks{eris.runa@mis.mpg.de}}
\affil{Max Planck Institut for Mathematics in the Sciences,\\ Inselstrasse 22, Leipzig\\ Germany}

\usepackage{esint}

\newcommand{\abs}[1]{{\lvert #1\rvert}}

\newcommand{\1}{{\mathchoice {1\mskip-4mu\mathrm l}      
      {1\mskip-4mu\mathrm l} 
      {1\mskip-4.5mu\mathrm l} {1\mskip-5mu\mathrm l}}}

\newcommand{\Acal}   {{\mathcal A }}

\newcommand{\Fcal}   {{\mathcal F }}

\newcommand{\Mcal}   {{\mathcal M }} 
\newcommand{\Ncal}   {{\mathcal N }}

\newcommand{\Ucal}   {{\mathcal U }} 
\newcommand{\Vcal}   {{\mathcal V }}

\def\Dscr{\mathscr{D}}

\usepackage{parskip}
\usepackage{titlesec}
\usepackage{titling}

\date{}

\usepackage{microtype}

\begin{document}

\maketitle

\begin{abstract}
   We study the limit of large volume equilibrium Gibbs measures for a rather general Hamiltonians.   
   In particular, we study Hamiltonians which arise in naturally in Nonlinear Elasticity and Hamiltonians (containing surface terms) which arises naturally in Fracture Mechanics.   
   In both of these settings we show that an integral representation holds for the limit.  
   Moreover, we also show a homogenization result for the Nonlinear Elasticity setting.  
   This extends a recent result of R.~Koteck\'y and S.~Luckhaus in \cite{LK-comm-math-phys}. 
\end{abstract}


\section{Introduction} 
\label{sec:representation-theorems-introduction}
Recently, R.~Koteck\'y and S.~Luckhaus, have shown a remarkable result. 
They prove that in a fairly general setting, the limit of large volume equilibrium Gibbs measures for elasticity type Hamiltonians with clamped boundary conditions. 
The ``zero''-temperature case was considered by R.~Alicandro and M.~Cicalese in \cite{MR2083851}. 

Let us now briefly explain the results contained in \cite{LK-comm-math-phys}. 

The space of microscopic configurations consists of all $\varphi:\varepsilon\Z^d\to\R^m$.
In particular, if one considers $m=d$ then this would model the elasticity situation. 
In this case $\varphi(x)$ can be interpreted as the displacement of the atom ``positioned in $x$''. 
When $m=1$, this would model the random surface case, where $\varphi(x)$ can be interpreted as the height. 

In order to define the Gibbs measure, we fix a configuration $\psi$, a set $A\subset \R^{d}$, an Hamiltonian $H$ and a finite range interaction $U$. 
Namely there exists a set $F\subset \Z^{d}$, such that $U: \R^{F}\to \R$ and one denotes $R_0$ the range of the potential $U$ \ie $R_{0}=\diam(F)$. 
Denote by $\varphi_{F}$ to be the restriction of $\varphi$ to $\varepsilon F$. For simplicity of notation we denote $A_{\varepsilon}:=A\cap \varepsilon \Z^{d}$. 
Then they define the Hamiltonian $H$ via
\begin{equation*}
   H_{A,\varepsilon}(\varphi)= \sum_{j\in\varepsilon\Z^d\colon \tau_j(F)\subset A } U(\varphi_{\tau_j(F)})
\end{equation*}
with $\tau_j(F)=\varepsilon F+j=\{i\colon i-j\in \varepsilon F\}$.

Finally denote by $\1_{A,\varepsilon,\psi}$ the indicator function of the set 
\begin{equation*}
   \{ \varphi\in (\R^m)^{A_{\varepsilon}}:\   |\varphi(i)- \psi(i)|<1  \text{ for all } i\in S_{R_0}(\Lambda)\},
\end{equation*}
where 
\begin{equation*} 
   \begin{split}
      S_{R_0}(A)=\{x\in A_{\varepsilon} |  \dist(x,\varepsilon \Z^d\setminus A_{\varepsilon})\leq\varepsilon R_0\}.
   \end{split}
\end{equation*} 
Then the Gibbs measure $\mu_{A,\psi}(\d\varphi)$ is a measure on $(\R^m)^{\varepsilon A}$ and  is defined by
\begin{equation*}
   \mu_{A,{\varepsilon},\psi}(\d \varphi)=\frac{\exp\big\{- \beta H_{A,\varepsilon}(\varphi)\big\}}{Z_{A,\varepsilon,\psi}}
   \1_{A, \varepsilon,\psi}(\varphi)\prod_{i\in A_{\varepsilon}} d \varphi(i)
\end{equation*}
where $Z_{{A,\varepsilon,\psi}}$ is such that the above is a probability measure. 

Moreover, they assume that

\begin{enumerate}
   \item[{(1)}] There exist constants $ p>0$ and  $c\in (0,\infty)$  such that 
      \begin{equation*}
         U(\varphi_F)\geq c|\nabla \varphi(0)|^p
      \end{equation*}
      for any  $\varphi\in (\R^m)^{\varepsilon\Z^d}$. 
   \item[{(2)}] There exist  $r>1$  and  $C\in(1,\infty)$  such that 
      \begin{equation*}
         U(s \varphi_A+(1-s) \psi_A + \eta_A)\leq C\big(1+U( \varphi_A) + U(\psi_A) +\sum_{i\in \varepsilon F} |\eta(i)|^{r}\big)
      \end{equation*}
      for any $s\in[0,1]$ and any   $\varphi,\psi,\eta\in (\R^m)^{\varepsilon\Z^d}$.
\end{enumerate}

With the above notation, in \cite{LK-comm-math-phys}, the following theorem is proved:

\begin{theorem} 
   \label{thm:kotecky-luckhaus-main}
   Let $U$ be as above with $\frac{1}{r}>\frac{1}{p}-\frac{1}{d}$. For every  $u\in W^{1,p}(\Omega)$, let us define 
   \begin{equation*}
      F_{\kappa,\varepsilon}(u)=- \varepsilon^d \abs{\Omega}^{-1} \log Z_{\Omega_{\varepsilon }}({\mathcal N}_{\Omega_\varepsilon ,r}(u,\kappa)),
   \end{equation*}
   and 
   \begin{equation*}
      \begin{split}
         F_{\kappa}^+(u)=\limsup_{\varepsilon\to 0} F_{\kappa,\varepsilon}(u)\\
         F_{\kappa}^-(u)=\liminf_{\varepsilon\to 0} F_{\kappa,\varepsilon}(u)
      \end{split}
   \end{equation*}
   Then, there exist $W$ quasi-convex such that the following  hold 
   \begin{enumerate} 
      \item $\lim_{\kappa\to 0} F_{\kappa}^-(u)\geq \frac{1}{|\Omega|}\int_{\Omega} W(\nabla u(x)) \dx$.
      \item  If $u\in W^{1,r}(\Omega)$ then $\lim_{\kappa\to 0} F_{\kappa}^+(u)\leq \frac{1}{|\Omega|}\int_{\Omega} W(\nabla u(x)) \dx$.
   \end{enumerate} 
\end{theorem}

They also give an explicit formula of $W$. Moreover, via an example, they show   that $W$ may eventually not be convex.

From the above result it is not very difficult to obtain a Large  Deviation principle.

The crucial step in the proof of Theorem~\ref{thm:kotecky-luckhaus-main} is based on the possibility to approximate with partition functions on cells of a triangulation given in terms
of $L^r$-neighbourhoods of linearizations of a minimiser of the rate
functional. An important tool that allows them to impose a boundary
condition on each cell of the triangulation consists in switching between the
corresponding partition
function $Z_{\Omega_\varepsilon }({\mathcal N}_{\Omega_\varepsilon ,r}(v,\kappa))$ and the version $Z_{\Omega_\varepsilon }({\mathcal N}_{\Omega_\varepsilon ,r}(v,2\kappa)\cap{\mathcal N}_{\Omega_\varepsilon ,R_0,\infty}(Z))$  with  an additional soft clamp $|\varphi(i)- \psi(i)|<1$ enforced  in the boundary strip of the width $R_0> \diam(A)$ with $Z\in  {\mathcal N}_{\Omega_\varepsilon ,r}(v,\kappa)$ arbitrarily chosen.

We improve their result in the following manner:
\begin{enumerate} 
   \item We consider Hamiltonians, where the interaction is \emph{not} of finite range and is dependent\footnote{for the precise definition see the next section}  on the scale $\varepsilon$ and the position $x$.  We are also able to give an homogenisation result.   
   \item By considering a different version of the interpolation argument we 
      are able to consider ``hard'' boundary condition instead of the clamped 
      ones. In our opinion, this type of boundary conditions are more in line 
      with the standard theory of Statistical Mechanics.   
   \item We simplify some of the arguments by relying on the representation 
      formulas, hence avoiding triangulation  arguments.  
   \item We are able to consider more general potentials, which ``relax'' in $\sbv$.  
\end{enumerate} 

\section{Sobolev Representation Theorems}
\label{sec:sobolev-representation}

\subsection{Preliminary results} 
Let $\Omega $ be an open set.  We denote by $\Acal(\Omega) $ the family of all open sets contained in $\Omega$. 
We now recall a well-known result in measure theory due to E.  De Giorgi and G.  Letta. 
The proof can be found in \cite{AFPBV}.

\begin{theorem} 
   \label{thm:degiorgiLetta}
   Let $X $ be a metric space and let us denote by $\Acal $ its open sets.  
   Let $\mu:\mathcal A\to [0,\infty] $ be an increasing set function such that 
   \begin{enumerate}[label=({DL}\arabic*)] 
      \item \label{cond:DG-L:insieme-vuoto} $\mu(\emptyset )=0 $;
      \item \label{cond:DG-L:subadditivita} $A,B\in \Acal $ then $\mu(A\cup B)\leq \mu(A) +\mu(B)$;
      \item $A,B\in \Acal $, such that $A\cap B =\emptyset $ then 
         $\mu(A\cap B)\geq \mu(A)+\mu(B)$
      \item $\mu(A)=\sup \insieme{\mu(B):\ B\Subset A} $.  
         Then, the extension of $\mu $ to  every $C\subset X $ given by 
         \begin{equation*} 
            \begin{split}
               \mu(C)=\inf \insieme{\mu(A):\ A\in\Acal,\ A \supset C}
            \end{split}
         \end{equation*} 
         is an outer measure.   In particular the restriction of $\mu $ to the Borel $\sigma $-algebra is a positive measure.  
   \end{enumerate} 
\end{theorem}

We recall the well-known integral representation formulas (see \cite{opac-b1087300}).
\begin{theorem}\label{thm:integral_represenatation_sobolev} 
   Let $1\leq p<\infty$ and let $F:W^{1,p}\times{\cal
      A}(\Omega)\to[0,+\infty]$ be a functional satisfying the following
   conditions:
   \begin{enumerate}
      \item[(i)]{\textrm{(locality)}} $F$ is local, i.e. $F(u,A)=F(v,A)$ if $u=v$ \ae on $A\in{\cal A}(\Omega)$;
      \item[(ii)]{\textrm{(measure property)}} for all $u\in W^{1,p}$ the set function
         $F(u,\cdot)$ is the restriction of a Borel measure to ${\cal
            A}(\Omega)$;
      \item[(iii)]{\textrm{(growth condition)}} there exists $c>0$ and $a\in L^1(\Omega)$ such
         that
         \begin{eqnarray*}
            F(u,A)\leq c\int_{A}(a(x)+|Du|^p)\ \dx
         \end{eqnarray*}
         for all $u\in W^{1,p}$ and $A\in{\cal A}(\Omega)$;
      \item[(iv)]{\textrm{(translation invariance in $u$)}} $F(u+z,A)=F(u,A)$ for all
         $z\in\R^d$, $u\in W^{1,p}$ and $A\in{\cal A}(\Omega)$;
      \item[(v)]{\textrm{(lower semicontinuity)}} for all $A\in{\cal A}(\Omega)$
         $F(\cdot,A)$ is sequentially lower semicontinuous with respect to
         the weak convergence in $ W^{1,p}$.
   \end{enumerate}

   Then there exists a Carath\'eodory function
   $f:\Omega\times\M^{d\times N}\to[0,+\infty)$ satisfying the growth
   condition
   \begin{eqnarray*}
      0\leq f(x,M)\leq c(a(x)+|M|^p)
   \end{eqnarray*}
   for all $x\in\Omega$ and $M\in M^{d\times N}$, such that
   \begin{eqnarray*}
      F(u,A)= \int\limits_{A}f(x,Du(x)) \dx
   \end{eqnarray*}
   for all $u\in W^{1,p}$ and $A\in{\cal A}(\Omega)$.\\
   If in addition it holds
   \begin{itemize}
      \item[(vi)]{\textrm{(translation invariance in $x$)}}
         \begin{eqnarray*}
            F(Mx,B(y,\ro))= F(Mx,B(z,\ro))
         \end{eqnarray*}
   \end{itemize}
   for all $M\in M^{d\times N},\ y,z\in\Omega$, and $\ro>0$ such that
   $B(y,\ro)\cup B(z,\ro)\subset\Omega$, then $f$ does not depend on
   $x$.
\end{theorem} 

\subsection{Hypothesis and Main Theorem} 
\label{subsec:hypothesis-and-main-sobolev}

For any $u\in L^1_{\text{loc}}(\R^d,\R^m)$,  let 
$X_{u,\varepsilon}: \Z^d\to\R^m$  and $\varphi:\varepsilon \Z^{d}\to \R^{m}$  be defined by  
\begin{equation}
   \label{eq:def-varphi-X}
   \begin{split}
      X_{u,\varepsilon}(i)&= \frac1{\varepsilon} \fint_{\varepsilon i + Q(\varepsilon)} u(y)\, \dy\\
      \varphi_{u,\varepsilon}(\varepsilon i)&= \frac1{\varepsilon} \fint_{\varepsilon i + Q(\varepsilon)} u(y)\, \dy
   \end{split}
\end{equation}
for any $i\in\Z^d$. Here, 
$Q(\varepsilon)=[-\tfrac{\varepsilon}2,\tfrac{\varepsilon}2]^d$ and $\fint$ 
denotes the mean value,\ie for every $f\in L^{1}(\R^{d})$
\begin{equation*} 
   \begin{split}
      \fint_{A}f(x)\dx = \frac{1}{| A |} \int_{A} f(x)\dx
   \end{split}
\end{equation*} 

Let $u\in W^{1,p}(\R^{d})$, $A$ is an open set and $p\geq 1$. Then it is not 
difficult to prove that 
\begin{equation} 
   \label{eq:converngence-of-discrete-sobolev-norm}
   \begin{split}
      \lim_{\varepsilon \downarrow 0}\sum_{x\in A_{\varepsilon }}\varepsilon 
      ^{d} |\nabla \varphi_{u} (x) |^{p} = \int_{A}|\nabla u | ^{p}.
   \end{split}
\end{equation} 

On the other hand, let 
\begin{equation}
   \Pi_{\varepsilon}: (\R^m)^{\Z^d}_0\to  W^{1,p}(\R^d)
\end{equation}
be a canonical  interpolation $X\to v$ such that $v({\varepsilon} i)=\varepsilon \varphi (\varepsilon i)$ for any $i\in\Z^d$.
Here, $(\R^m)^{\Z^d}_0$ is the set of functions $X:\Z^d\to\R^m$ with finite support.
To fix ideas, we can consider a triangulation of $\Z^d$ into simplexes with vertices in $\varepsilon\Z^d$, and choose $v$ on each simplex as the linear interpolation of the values ${\varepsilon} \varphi(\varepsilon i)$ on the vertices $\varepsilon i$.

Let $\Omega $ be an open set with regular boundary.  We denote by 
$\Omega_{\varepsilon}= \varepsilon \Z^{d}\cap \Omega$
and by $\Acal(\Omega) $ the set of all open sets contained in $\Omega $ with 
regular boundary.  For every set $A\in \Acal(\Omega) $, we define  

\begin{equation*} 
   \begin{split}
      R_{\varepsilon}^{\xi} ( A ):=\{\alpha\in\varepsilon\Z^d\ [\alpha,\alpha+\varepsilon\xi]\subset A \}, 
   \end{split}
\end{equation*} 
where by $[x,y]$ we mean the segment connecting $x$ and $y$, \ie $\{ \lambda x +(1-\lambda )y:\ \lambda \in [0,1]\}$.

The Hamilton $H$ is defined by
\begin{equation*} 
   \label{eq:hamiltoniana}
   \begin{split}
      H(\varphi,\varepsilon) := \sum_{\xi\in \Z^d}\sum_{x\in R^{\xi}_{\varepsilon}(\Omega)}
      f_{\xi,\varepsilon}(x,\nabla_{\xi}\varphi),
   \end{split}
\end{equation*} 
where $\xi\in \Z^{d}$,  and
\begin{equation*} 
   \begin{split}
      \nabla_{\xi } \varphi (x) := \frac{\varphi (x +\varepsilon \xi ) -\varphi (x)}{|\xi  |}.
   \end{split}
\end{equation*} 
We will also define a second Hamilton which takes into account also the boundary condition:
\begin{equation*} 
   \begin{split}
      H_{\infty}({\varphi},A,\varepsilon):=\sum_{\xi\in\Z^{d}}\sum_{x\in A_{\varepsilon}} f_{\varepsilon ,\xi }(x,\nabla _{\xi }\varphi(x)).
   \end{split}
\end{equation*}

The functions $f_{\xi,\varepsilon}$ will be specified later.

In order to apply the representation formulas, we shall need to localize. 
For this reason,  for every $\varepsilon> 0 $ and $ A \subset \Omega $ open, set we introduce 
\begin{equation*} 
   \begin{split}
      H(\varphi,A,\varepsilon):=  \sum_{\xi\in \Z^d}\sum_{x\in R^{\xi}_{\varepsilon}(A)} 
      f_{\xi,\varepsilon}\big(x,\nabla_{\xi}\varphi (x)\big).
   \end{split}
\end{equation*} 

For simplicity of notation, we will also denote
\begin{equation*} 
   \begin{split}
      H^{\xi}(\varphi,A,\varepsilon):= \sum_{x\in R^{\xi}_{\varepsilon}(A)} f_{\xi,\varepsilon}\big(x,\nabla_{\xi}\varphi (x)\big).
   \end{split}
\end{equation*} 

Moreover, let $\insieme{ e_1,\ldots,e_d  }$ be the standard basis of $\R^{d}$.
In this section, the functions $ f_{\xi,\varepsilon}$ will satisfy the followings
\begin{enumerate}[label=({C}\arabic*)  ] 
   \item \label{cond:positivita} $f_{\xi,\varepsilon}> 0 $;
   \item \label{cond:stima-alto} there exist constants $C_\xi $ such that
      \begin{equation*} 
         \begin{split}
            f_{\xi,\varepsilon}(x,s + t) \leq f_{\xi,\varepsilon}(x,s) +C_\xi(|t |^{p} +1);
         \end{split}
      \end{equation*} 
      where the constants $C_\xi $ satisfy
      \begin{equation*} 
         \begin{split}
            \sum_{\xi\in \Z^{d}}  C_\xi< +\infty;
         \end{split}
      \end{equation*} 
   \item \label{cond:primi-vicini} there exists a constant $C$ such that
      \begin{equation*} 
         \begin{split}
            f_{e_i,\varepsilon}(x,t) \geq C \max(|t |^p -1,0).
         \end{split}
      \end{equation*} 
\end{enumerate}

\begin{remark}
   The above conditions are natural in the Calculus of Variations. Very similar conditions appear also in \cite{MR2083851}. To my knowledge, they are the most general conditions for a representation formula in the discrete to continuum limit. 
\end{remark}

In order to simplify notation, throughout this paper, we  will denote by 
\begin{equation}
   \label{dVarphi}
   \d\varphi = \Pi_{x\in A_{\varepsilon}} \d\varphi(x),
\end{equation}
whenever it is clear from the context.

For every $A\in\Acal(\Omega)$, we define the free-energy as
\begin{equation} 
   \label{eq:def-free-energy}
   \begin{split}
      F(u,A,\kappa,\varepsilon) := - \frac{\varepsilon^{d}}{|A |} \log \int_{\Vcal(u,A,\kappa,\varepsilon )} \exp \Big(- H(\varphi,A, \varepsilon) \Big)\d\varphi\\
      F_{\infty }(u,A,\kappa,\varepsilon) := -\varepsilon ^{d} \log \int_{\Vcal_{\infty}(u,A,\kappa,\varepsilon)} \exp\Big(-H_{\infty}(\varphi,A, \varepsilon) \Big)\d\varphi,
   \end{split}
\end{equation} 
where 
\begin{equation*} 
   \begin{split}
      \Vcal(u,A,\kappa,\varepsilon  )&= \insieme{\varphi:A_{\varepsilon }\to  \R^m|\ \frac{\varepsilon ^{d}}{|A| 
            ^{d}}\sum_{x\in A_{\varepsilon}}|u-\varepsilon \varphi|^{p}\leq \kappa^{p} }\\
      \Vcal_{\infty }(u,A,\kappa,\varepsilon  )&= \insieme{\varphi:\varepsilon\Z^{d} \to  \R^m|\ \frac{\varepsilon ^{d}}{|A| 
            ^{d}}\sum_{x\in A_{\varepsilon}}|u-\varepsilon \varphi|^{p}\leq 
         \kappa^{p}, \text{ and } \varphi(x) = \varphi_{u,\varepsilon }(x)\ \forall x\not \in A_{\varepsilon} },
   \end{split}
\end{equation*} 
where $\varphi_{u,\varepsilon }$ is defined in \eqref{eq:def-varphi-X}.

\begin{remark}
   \label{rmk:spf1}
   From the physical point of view in the above expression it is natural to pass to the limit first for $\varepsilon\downarrow 0 $ and afterwords for $\kappa\downarrow 0$.  This is due to the close relation of the free-energy to the Large Deviation Principle. Due to the generality of our formulation it is not true in general that the above described limits exist. 
\end{remark}

In order to understand weather the limit first in $\varepsilon\downarrow 0$ and afterwords in $\kappa\downarrow 0$ (or up to sub-sequences) exists it is natural to introduce the following notations:
\begin{equation} 
   \label{eq:definizioni-F}
   \begin{split}
      F' (u,A,\kappa)&:= \liminf_{\varepsilon\downarrow 0} F(u,A,\kappa,\varepsilon)\\
      F''(u,A,\kappa)&:= \limsup_{\varepsilon\downarrow 0} F(u,A,\kappa,\varepsilon)\\
      F' (u,A)&:= \lim_{\kappa\downarrow 0}\liminf_{\varepsilon\downarrow 0} 
      F(u,A,\kappa,\varepsilon)=\lim_{\kappa\downarrow 0}F' (u,A,\kappa)\\
      F''(u,A)&:= \lim_{\kappa\downarrow 0}\limsup_{\varepsilon\downarrow 0} 
      F(u,A,\kappa,\varepsilon)=\lim_{\kappa\downarrow 0}F''(u,A,\kappa) \\
      F_{\infty }' (u,A,\kappa)&:= \liminf_{\varepsilon\downarrow 0} F_{\infty }(u,A,\kappa,\varepsilon)\\
      F_{\infty }''(u,A,\kappa)&:= \limsup_{\varepsilon\downarrow 0} F_{\infty }(u,A,\kappa,\varepsilon)\\
      F_{\infty }' (u,A)&:= \lim_{\kappa\downarrow 0}\liminf_{\varepsilon\downarrow 0} 
      F_{\infty }(u,A,\kappa,\varepsilon)=\lim_{\kappa\downarrow 0}F_{\infty }' (u,A,\kappa)\\
      F_{\infty }''(u,A)&:= \lim_{\kappa\downarrow 0}\limsup_{\varepsilon\downarrow 0} 
      F_{\infty }(u,A,\kappa,\varepsilon)=\lim_{\kappa\downarrow 0}F_\infty ''(u,A,\kappa) \\
   \end{split}
\end{equation}

One of the main steps will be to show that $F'_{\infty }=F'$ and that $F''_{\infty }=F''$. 
The basic intuition behind is the so called interpolation lemma, which is well-known in the $\Gamma$-convergence community. 
Very informally, what it says that if one imposes ``closeness'' in $L^{p}(A)$ to some regular function $u$, then one can also impose the boundary condition by ``paying a very small price in energy''. 
More precisely, given a sequence $\{ v_{n}\}$ such that $v_{n}\to u$ in $L^{p}(A)$, where $A$ is an open set, then there exists a sequence $\{\tilde{v}_{n}\}$, such that $\tilde{v}_n \to u$ in $L^{p}(A)$, $\tilde{v}_{n}|_{\partial \Omega} =u|_{\partial \Omega }$  and such that
\begin{equation*} 
   \begin{split}
      \liminf_{n} \int_A |\nabla \tilde{v}_{n}|^{2} \leq \liminf_{n} \int_A |\nabla {v}_{n}|^{2}.
   \end{split}
\end{equation*} 

\begin{rmk} 
   \label{rmk:simple1}
   \begin{enumerate} 
      \item 
         The functional $F(u,A,\kappa,\varepsilon) $ is monotonically decreasing in $\delta , \kappa> 0 $, i.e.  
         \begin{equation*} 
            \begin{split}
               F(u,A,\kappa,\varepsilon) \leq F(u,A,\kappa+\delta,\varepsilon ).
            \end{split}
         \end{equation*} 
         This justifies the outer limit in the formulas of \eqref{eq:definizioni-F}. 
         Moreover, the outer limit in the formulas in  in  \eqref{eq:definizioni-F} can be 
         substituted with the supremum \ie
         \begin{equation*} 
            \begin{split}
               F' (u,A)&:= \sup_{\kappa>0}\liminf_{\varepsilon\downarrow 0} F(u,A,\kappa,\varepsilon)=\sup_{\kappa >0}F' (u,A,\kappa),\\
               F''(u,A)&:= \sup_{\kappa>0}\limsup_{\varepsilon\downarrow 0} 
               F(u,A,\kappa,\varepsilon)=\sup_{\kappa >0}F''(u,A,\kappa) .
            \end{split}
         \end{equation*} 
      \item Let $A,B $ be two open sets such that $A\cap B =\emptyset$, then 
         from the definitions it is not difficult to prove that  
         \begin{equation*} 
            \begin{split}
               F'(u,A)+F'(u,B) = F'(u,A\cup B)\qquad \text{and}\qquad F''(u,A)+F''(u,B) = F''(u,A\cup B).
            \end{split}
         \end{equation*} 
      \item Whenever the function $u$ is linear and the functions $f_{\xi,\varepsilon}$ do not depend on $\varepsilon$ and the space variable $x$, it is well-known that $F' =F''$.
         In Theorem~\ref{thm:sobolev_homog_main}, we are going to prove a more general result, which contains as a particular case the previous claim.
   \end{enumerate} 
\end{rmk}

\begin{prop} 
   \label{prop:propriet-compact}
   The maps $F',F'' $ are lower semicontinuous with respect to the $L^{p}(A) $ convergence. 
   Moreover, there exists a sequence $\insieme{\varepsilon_{n}} $ such that
   \begin{equation} 
      \label{eq:propret-compat-1}
      \begin{split}
         F'_{\insieme{\varepsilon_{n}}}(u)= F''_{\insieme{\varepsilon_{n}}}(u),
      \end{split}
   \end{equation} 
   where 
   \begin{equation*} 
      \begin{split}
         F'_{\insieme{\varepsilon_{n}} }(u):= \lim_{\kappa \downarrow 0} \liminf_{n\to \infty} 
         F(u,A,\kappa,\varepsilon_{n})  \qquad\text{and} \qquad F''_{\insieme{\varepsilon_{n}}}(u):=\lim_{\kappa 
            \downarrow 0} \limsup_{n\to \infty} F(u,A,\kappa,\varepsilon_{n}) .
      \end{split}
   \end{equation*} 

\end{prop}

\begin{proof} 
   Let us recall the notations
   \begin{equation*} 
      \begin{split}
         F'_{\{ \varepsilon_{k}\}}(u,A,\kappa)= \liminf_{n\to +\infty}F(u,A,\kappa,\varepsilon_{n})\quad \text{and}
                         \quad
         F''_{\{ \varepsilon_{k}\}}(u,A,\kappa)= \limsup_{n\to +\infty}F(u,A,\kappa,\varepsilon_{n}).\\ 
      \end{split}
   \end{equation*} 

   Using $F(v,A,\kappa,\varepsilon ) \geq F (u,A,\kappa +\delta,\varepsilon) $ where $\|u-v \|_{L^{p}(A)}<\delta $, one has that
   \begin{equation*} 
      \begin{split}
         F'(v,A,\kappa)=\liminf_{n\to \infty} F (u,A,\kappa ,\varepsilon_{n}) \geq
         \liminf_{n\to \infty} F (v,A,\kappa +\delta,\varepsilon_{n}) =F'(u,A,\kappa +\delta ).
      \end{split}
   \end{equation*} 
   Thus, 
   \begin{equation*} 
      \begin{split}
         \liminf_{v\to u} \sup_{\kappa >0} F'(u,A,\kappa) \geq \sup_{\kappa>0} 
         F'(u,A,\kappa+\delta)
      \end{split}
   \end{equation*} 
   and finally passing also to the supremum in $\delta $ one has that $F'$ is 
   lower semicontinuous.  The statement for $F''$ follows in a similar fashion.

   Fix $\Dscr $ a countable dense set in $L^{p}(A) $ and let $\Ucal $ be the set of all balls centered in the elements of $\Dscr$ with  radii in $[0,1]\cap \Q$.  
   Let us enumerate the balls in $\Ucal$, namely $\Ucal:=\{ B_{i}:\ i\in \N\}$. 

   Let $u_{1}\in B_{1}$ be such that  
   \begin{equation*} 
      \begin{split}
         F'(u_{1},A) \leq \inf_{u\in B_{1}}F'(u) + \diam(B_{1}). 
      \end{split}
   \end{equation*} 
   Let $\{ \varepsilon^{(1)}_{n}\}$ be the sequence such that 
   \begin{equation*} 
      \begin{split}
         F'(u_{1},A)=F''(u_{1},A)=\lim_{\kappa \downarrow 0}\lim_{n\to  \infty }F(u_{1},A,\kappa ,\varepsilon_{n}^{(1)}).
      \end{split}
   \end{equation*} 
   In a similar way as for $B_{1}$, let $u_{2}\in B_{2}$ be such that
   \begin{equation*} 
      \begin{split}
         F'_{\{ \varepsilon_{n}^{(1)}\}}(u_{2},A) \leq \inf_{u\in B_{2}}F'_{\{ \varepsilon_{n}^{(1)}\}}(u) + \diam(B_{2}).
      \end{split}
   \end{equation*} 
   Moreover, let $\{ \varepsilon ^{(2)}_{n}\}\subset \{ \varepsilon ^{(1)}_{n}\} $ be such that 
   \begin{equation*} 
      \begin{split}
         F'(u_{2},A)=\lim_{\kappa \downarrow 0}\lim_{n\to  \infty }F(u,A,\kappa ,\varepsilon_{n}^{(2)}).
      \end{split}
   \end{equation*} 
   By an induction procedure it is possible to produce  a subsequence 
   $\{ \varepsilon ^{(k+1)}_{n}\}\subset\{\varepsilon ^{(k)}_{n}\} $ such that 
   \begin{equation*} 
      \begin{split}
         F'(u_{k},A)=\lim_{\kappa \downarrow 0}\lim_{n\to  \infty }F(u_{k},A,\kappa ,\varepsilon_{n}^{(k)}),
      \end{split}
   \end{equation*} 
   where $u_{k}$ is chosen such that 
   \begin{equation*} 
      \begin{split}
         F'_{\{ \varepsilon_{n}^{(k+1)}\}}(u_{k+1},A) \leq \inf_{B_{k+1}}F'_{\{ \varepsilon_{n}^{(k)}\}} + \diam(B_{k+1}).
      \end{split}
   \end{equation*} 

   By a diagonal argument it is possible to chose a single sequence $\{ \varepsilon_k\}$, such that all the above are satisfied. 
   Because the second claim of the Proposition~\ref{prop:propriet-compact} consists in showing \eqref{eq:propret-compat-1} for a particular sequence, one can assume without loss of generality that it satisfies the above relations. 

   Let us now show that $F'_{\{ \varepsilon_{n}\}}=F''_{\{ \varepsilon_{n}\}}$. 
   From the definitions it is trivial that $F'_{\{ \varepsilon_{n}\}} \leq F''_{\{ \varepsilon_{n}\}}$. 
   Let us now show the opposite inequality. Fix $u$. 
   For every $i$ such that $u\in B_{i}$  we have that\footnote{by the above construction}
   \begin{equation*} 
      \begin{split}
         F'_{\{ \varepsilon_{n}\}}(u,A) + \diam(B_{i})\geq F'_{\{ \varepsilon_{n}\}}(u_{i},A) = F''_{\{ \varepsilon_{n}\}}(u_{i},A).
      \end{split}
   \end{equation*} 
   Passing to the limit for $i\to  \infty $, using the lower semicontinuity of $F''_{\{ \varepsilon_{n}\}}$, and the arbitrarity of $\diam(B_i)$, we have the desired result. 
\end{proof}

Fix $\Omega$ an open set, $\varepsilon >0$ and $u\in W^{1,p}(\R^{d})$ and let $\varphi_{u,\varepsilon }$ be defined by in \eqref{eq:def-varphi-X}.  

We are now able to write the main result in this section:
\begin{theorem} 
   \label{thm:main-representation-sobolev-statmech}
   Assume the above hypothesis.   Then for every infinitesimal sequence 
   $(\varepsilon_{n})$ there exists a subsequence $(\varepsilon_{n_k})$ and 
   there exists a function $W:\Omega \times \R^{d\times m} \to \R $ (depending 
   on $\{\varepsilon_{n_k}\}$) such that
   \begin{equation} 
      \label{eq:thm-main-representation-sobolev-statmech-1}
      \begin{split}
         F'_{\{\varepsilon_{n_k}\}}(u,A)= F''_{\{ \varepsilon_{n_k} \}}(u,A) =  \int_A 
         W(x,\nabla u) \dx.  
      \end{split}
   \end{equation} 
\end{theorem}

\subsection{Proofs} 

The next technical lemma asserts that finite difference quotients along any direction can be controlled by finite difference quotients along the coordinate directions:

\begin{lemma}[{\cite[Lemma~3.6]{MR2083851}}]
   \label{lemma:zig-zag-sobolev} 
   Let $A\in {\Acal}(\Omega)$ and set $A_\varepsilon=\{x\in A:\ \dist(x,\partial A)>2\sqrt {d}\varepsilon\}$. Then for any $\xi\in\Z^{d}$ and
   $\varphi: A_\varepsilon \to \R^{m}$, it holds
   \begin{eqnarray}\label{stimadiff}
      \sum_{x\in R^{\xi}_{\varepsilon}(A)} \Big|\frac{\varphi(x+\varepsilon\xi) -\varphi(x)}{|\xi|}\Big|^p\leq C
      \sum_{i=1}^d\sum_{x\in R^{e_i}_{\varepsilon}(A)}|\nabla_{i}\varphi(x)|^p,
   \end{eqnarray}
   where the constant $C $ is independent of $\xi $.  
\end{lemma}

%
%
%

The following lemma is a simple modification of \cite[Lemma~A.1]{LK-comm-math-phys}:

\begin{lemma}
   \label{lemma:kotecky-luckhaus_A_1}
   Let $A\subset \R^d$ and let $g:\R\to [0,+\infty]$ be such that  
   \begin{equation*}
	   \int_{\R} \exp(-g(t)) \dt =:c.
   \end{equation*}
   Then there exists $\varepsilon_{0}$ such that for every $\varepsilon \leq \varepsilon_{0}$, it holds
   \begin{equation*} 
      \begin{split}
         \int_{\Vcal_{\infty}(0,A,\kappa,\varepsilon )} \exp( -\sum_{x\in A_{\varepsilon} }\sum_{i=1}^{d} g(|\nabla_{i} \varphi(x)|) ) \d\varphi \leq     \exp\Big(C|A_{\varepsilon}|\log(c)\Big).  
      \end{split}
   \end{equation*} 
   Moreover,
   \begin{equation*} 
      \begin{split}
         \int_{\Vcal(0,A,\kappa,\varepsilon )} \exp( -\sum_{i=1}^{d}\sum_{R^{e_i}_{\varepsilon}} g(|\nabla_{i} \varphi(x)|) ) \d\varphi \leq  \frac{C|A_{\varepsilon} |}{\varepsilon}    \exp\Big(C|A_{\varepsilon}|\log(c)\Big).  
      \end{split}
   \end{equation*} 
\end{lemma}

\begin{proof} 
   The second statement is contained in \cite[Lemma~A.1]{LK-comm-math-phys}. 
   We will show now how to obtain the first statement which is also a simple modification of  \cite[Lemma~A.1]{LK-comm-math-phys}. 
   To see how to get the above statement it is enough to find a connected curve $\Gamma$ such that visits all the points exactly once, whose initial point is outside of $A_{\varepsilon}$ and afterwards never leaves $A_{\varepsilon}$. As this curve can be seen as a graph we have that
   \begin{equation*}
      \sum_{x\in A_{\varepsilon}} \sum_{i=1} ^d g(\nabla_{i} \varphi) \geq \sum_{e \in \mathrm{Edg}(\Gamma)} g(|\nabla_e \varphi|)
   \end{equation*}
   By using the above we have that
   \begin{equation*}
      \int_{\Vcal(0,A,\kappa,\varepsilon )} \exp( -\sum_{x\in A_{\varepsilon} }\sum_{i=1}^{d} g(|\nabla_{i} \varphi(x)|) ) \d\varphi \leq
      \int_{\R^{|\Gamma|}} \exp\big( -\sum_{e\in \mathrm{Edges}(\Gamma) } g(|\nabla_{e}\varphi| )\big) \d\varphi
   \end{equation*}
   from which the claim follows. 
\end{proof} 

Let $G^{\lambda}$ be the free-energy (see \eqref{eq:def-free-energy} for the definition) induced by the Hamiltonian
\begin{equation*} 
   \begin{split}
      \tilde{H}^{\lambda}(\varphi,A,\varepsilon):=\lambda \sum_{i=1}^{d}\sum_{x\in R_{\varepsilon }^{e_{i}}(A)} |\nabla_{i}\varphi |^{p}.
   \end{split}
\end{equation*}

\begin{lemma} 
   \label{lemma:utile-stime-g}
   Let $A$ be an open set.  Then, there exist constants $C_{\lambda},D_{\lambda}$, such that it holds
   \begin{equation*} 
      \begin{split}
         |A| C_{\lambda } \leq  \liminf_{\varepsilon\downarrow 0}G^{\lambda }(0,A,\kappa ,\varepsilon ) \leq \limsup_{\varepsilon\downarrow 0}G^{\lambda }(0,A,\kappa ,\varepsilon )  \leq D_{\lambda } |A|
      \end{split}
   \end{equation*} 
\end{lemma} 

\begin{proof} 
   Let us prove now the upper bound, namely
   \begin{equation} 
      \label{stima-alto-g}
      \begin{split}
         G^{\lambda }(0,A,\kappa ,\varepsilon ) \leq D_{\lambda } |A |.
      \end{split}
   \end{equation} 
   Let us  observe  that
   \begin{equation} 
      \label{eq:04181397841972}
      \begin{split}
         \tilde{H}^{\lambda}(\varphi ,A,\varepsilon )\leq 2^{p-1}d\cdot\lambda \sum_{x\in A_{\varepsilon}} |\varphi (x)|^{p},
      \end{split}
   \end{equation} 
   hence
   \begin{equation*} 
      \begin{split}
         \int_{\Vcal(0,A,\kappa ,\varepsilon )} \exp\left(-\tilde{H}^{\lambda }(\varphi ,A,\varepsilon )\right) \geq \int_{\{ \varphi :\ |\varepsilon \varphi | \leq \kappa \}} \exp\left( -\sum_{x\in A_{\varepsilon}}|\varphi (x)|^{p}\right).
      \end{split}
   \end{equation*} 
   Thus, by using the Fubini Theorem,  we have that
   \begin{equation*} 
      \begin{split}
         \int_{\Vcal(0,A,\kappa ,\varepsilon )} \exp\left(-\tilde{H}^{\lambda }(\varphi ,A,\varepsilon )\right) \geq \exp\left( -\varepsilon ^{-d} D_{\varepsilon,\varepsilon} |A|\right),
      \end{split}
   \end{equation*} 
   where
   \begin{equation*} 
      \begin{split}
         D_{\varepsilon,\lambda}:= - \log \int_{|t|\leq \kappa/\varepsilon}\exp\left( -t^{p}\right).
      \end{split}
   \end{equation*} 
   Using the definition of the free-energy, and the fact that $D_{\varepsilon,\lambda} \to - \log \int_{\R}\exp\left( -t^{p}\right)$, we have the first part of the claim. 

   The second inequality is implied by Lemma~\ref{lemma:kotecky-luckhaus_A_1}.
\end{proof}

\begin{lemma} 
   \label{lemma:stima-basso}
   Let $\insieme{f_{\xi,\varepsilon }} $ satisfy our hypothesis.
   Then there exists a constant $D$ such that for every $\kappa < 1$, one has that
   \begin{equation} 
      \label{eq:stima-basso-1}
      \begin{split}
         \exp\left(-\varepsilon ^{-d} F(u,A,\kappa ,\varepsilon )\right) \leq  
         \exp\left( D |A| \varepsilon ^{-d} + D 
            \sum_{i=1}^{d}\sum_{x\in R^{e_{i}}_{\varepsilon }(A)} |\nabla_{e_i}\varphi_{u,\varepsilon } (x) |^{p}  \right),
      \end{split}
   \end{equation} 
   where $\varphi_{u,\varepsilon }$.
\end{lemma} 

\begin{proof} 
   Given that $\|b-a\|^{p} \geq 2^{1-p}\|a \|^{p} - \|b \|^{p} $ one has 
   that there exists a constant $C_1$  such that 
   \begin{equation*} 
      \begin{split}
         H(\varphi ,A,\varepsilon ) \geq C_{1}\sum_{i=1}^{d}\sum_{x\in R^{e_{i} 
            }_{\varepsilon }{(A)}} |\nabla_{e_i}\varphi(x)  |^{p} \geq 
         C_{1} \sum_{i=1}^{d}\sum_{x\in R^{e_{i}  }_{\varepsilon }{(A)}} 
         |\nabla \psi |^{p} -C_{1} \sum_{i=1}^{d}\sum_{x\in R^{e_{i} 
            }_{\varepsilon }{(A)}} |\nabla_{e_i}(\varphi_{ u,\varepsilon} )(x)  |^{p},
      \end{split}
   \end{equation*} 
   where $\psi = \varphi - \varphi_{u,\varepsilon}$ and $\varphi_{u,\varepsilon}$ is defined in \eqref{eq:def-varphi-X}. 
   Hence, the estimate \eqref{eq:stima-basso-1} reduces to prove that there exists a constant $D$ such that 
   \begin{equation*} 
      \begin{split}
         \int_{\Vcal(0,A,\kappa ,\varepsilon )} \exp\Big(- C 
         \sum_{i=1}^{d}\sum_{x\in R^{e_{i} }_{\varepsilon }{(A)}} |\nabla_{e_i}\varphi |^p\Big) \leq 
         \exp\left( D|A| \varepsilon ^{-d}\right).
      \end{split}
   \end{equation*} 
   The above inequality was proved in Lemma~\ref{lemma:utile-stime-g}.

\end{proof}

\begin{remark} 
   \label{rmk:almost-monotonicity}
   A simple consequence of the reasoning done in Lemma~\ref{lemma:stima-basso}, is that there exists a constant $C$ such that 
   \begin{equation*} 
      \begin{split}
         A\mapsto F'(u,A) + C(\|\nabla u \|^p_{L^{p}(A)} +|A| ) \qquad  A\mapsto F''(u,A)+ C(\|\nabla u \|^p_{L^{p}(A)} +|A| )
      \end{split}
   \end{equation*} 
   are monotone with respect to the inclusion relation \ie for every $A\subset B$  it holds that 
   \begin{equation} 
      \label{eq:monotonicity_hypthesis}
      \begin{split}
         F'(u,A) + C(\|\nabla u \|^p_{L^{p}(A)} +|A| ) \leq F'(u,B) + C(\|\nabla u \|^p_{L^{p}(B)} +|B|).
      \end{split}
   \end{equation} 

   Without loss of generality we may assume that $F'$ and $F''$ satisfy \eqref{eq:monotonicity_hypthesis}. Indeed, there is a representation formula as in the claim of Theorem~\ref{thm:main-representation-sobolev-statmech} for 
   $$F_{\{ \varepsilon_{n_k}\}}(u,A) + \|\nabla u \|^p_{L^p(A)} + |A|,$$
   if and only if there is a representation formula for $F_{\{ \varepsilon_{n_k}\}}(u,A)$. 

   To make analysis less technical, we will assume \eqref{eq:monotonicity_hypthesis} will be assumed. 
   
\end{remark}

\begin{lemma} 
   \label{lemma:stima-alto}
   Let $f_{\xi ,\varepsilon }$ satisfy our hypothesis and let $A $ be an open 
   set. Then there exists  a constant $D>0$, such that 
   \begin{equation} 
      \label{eq:stima-alto-1}
      \begin{split}
         \exp\left(-\varepsilon ^{-d} F(u,A,\kappa ,\varepsilon )\right) \geq  
         \exp\left( - D |A |\varepsilon ^{-d} -  \sum_{i=1}^{d}\sum_{x\in R^{e_{i}}_{\varepsilon }(A)} |\nabla_{e_i}\varphi_{u,\varepsilon } (x) |  \right),
      \end{split}
   \end{equation} 
   where $\varphi_{u,\varepsilon }$ is defined in \eqref{eq:def-varphi-X}.
   Moreover, there exits $\varepsilon_{0},D_{0}$ such that  for every $\varepsilon <  \varepsilon_{0}$ one has that
   \begin{equation} 
      \label{eq:stima-alto-2}
      \begin{split}
         \exp\left(-\varepsilon ^{-d} F_{\infty}(u,A,\kappa ,\varepsilon )\right) \geq  
         \exp\left( - D_{0} |A |\varepsilon ^{-d} - D_{0} \sum_{i=1}^{d}\sum_{x\in R^{e_{i}}_{\varepsilon }(A)} |\nabla_{e_i}\varphi_{u,\varepsilon } (x) |  \right),
      \end{split}
   \end{equation} 
\end{lemma} 

\begin{proof} 
   Let us initially prove \eqref{eq:stima-alto-1}.
   Using Lemma~\ref{lemma:zig-zag-sobolev}, one has that there exists a 
   constant $C$ such that 
   \begin{equation*} 
      \begin{split}
         H(\varphi ,A,\varepsilon ) \leq C \sum_{i=1}^{d}\sum_{x\in R^{e_{i} }_{\varepsilon }{(A)}} |\nabla_{e_i}\varphi(x)  |^{p}
      \end{split}
   \end{equation*} 

   Given that $\|a + b \|^{p} \leq 2^{p-1} \|a \|^{p} + 2^{p-1}\|b \|^{p-1} $, there 
   exist a constant $C_{1} $  such that
   \begin{equation*} 
      \begin{split}
         H(\varphi,A,\varepsilon)\leq 
         C_{1} \sum_{i=1}^{d}\sum_{x\in R^{e_i}_{\varepsilon }(A)} (  |\nabla_{e_i}\varphi_{u,\varepsilon } |^{p}  +1 ) +\sum_{i=1}^{d}\sum_{x\in R^{e_i}_{\varepsilon }(A)}  |\nabla_{i} \psi(x) |^{p },
      \end{split}
   \end{equation*} 
   where $\psi = \varphi - \varphi_{u,\varepsilon }  $. 
   Hence, the estimate \eqref{eq:stima-alto-1} reduces to prove that there exists a constant $D$ such that 
   \begin{equation*} 
      \begin{split}
         \int_{\Vcal(0,A,\kappa ,\varepsilon )} \exp\Big(- C 
         \sum_{i=1}^{d}\sum_{x\in R^{e_{i} }_{\varepsilon }{(A)}} |\nabla_{e_i}\varphi |^p\Big) \leq 
         \exp\left( D|A| \varepsilon ^{-d}\right).
      \end{split}
   \end{equation*} 
   The above inequality was proved in Lemma~\ref{lemma:utile-stime-g}.

   Let us now turn to \eqref{eq:stima-alto-2}.  Let $A_1 \supset A$. 

   Given that $\|a + b \|^{p} \leq 2^{p-1} \|a \|^{p} + 2^{p-1}\|b \|^{p-1} $, there 
   exist a constant $C_{1} $  such that
   \begin{equation*} 
      \begin{split}
          H(\varphi,A_1,\varepsilon)\leq C_{1}  \sum_{\xi\in \Z^{d}}C_{\xi}\sum_{x\in A_{1,\varepsilon}} 
          |\nabla_{\xi}\psi |^{p} +  C_{1} \sum_{\xi\in \Z^{d}} \sum_{x\in A_{1,\varepsilon}} C_{\xi} (|\nabla_{\xi}\varphi_{u,\varepsilon} |^{p} +1 )
      \end{split}
   \end{equation*} 
   where $\psi = \varphi - \varphi_{u,\varepsilon }  $ and $\psi$ is $0$ outside $A_{\varepsilon}$. 
   By using the Lemma~\ref{lemma:zig-zag-sobolev}, and because $\psi$ is $0$ outside $A_{\varepsilon}$, we have that 
   \begin{equation*}
      \sum_{\xi\in \Z^d}\sum_{x\in A_{1,\varepsilon}} C_{\xi}  |\nabla_{\xi}\psi |^{p}  \leq  \sum_{x\in A_{1,\varepsilon}} \sum_{i=1}^{d}   |\nabla_{i}\psi |^{p}  = 
       \sum_{x\in A_{\varepsilon}} \sum_{i=1}^{d}   |\nabla_{i}\psi |^{p}  .
   \end{equation*}
   On the other side as, since $\varphi_{u,\varepsilon}$ is a discretization of $u$, it is not difficult to show that there exists an $\varepsilon_{0}$ an $C$ such that 
   \begin{equation}
      \sum_{\xi\in \Z^{d}}C_{\xi}\sum_{x\in A_{1,\varepsilon}}  |\nabla_{\xi}\varphi_{u,\varepsilon} |^{p} \leq C \sum_{x\in A_{\varepsilon}} \sum_{i=1}^{d}| \nabla_{i}\varphi_{u,\varepsilon}|
   \end{equation}

   Hence, the estimate \eqref{eq:stima-alto-2} reduces to prove that there exists a constant $D$ such that 
   \begin{equation*} 
      \begin{split}
         \int_{\Vcal(0,A,\kappa ,\varepsilon )} \exp\Big(- C 
         \sum_{i=1}^{d}\sum_{x\in A_{\varepsilon} } |\nabla_{e_i}\varphi |^p\Big) \leq 
         \exp\left( D|A| \varepsilon ^{-d}\right).
      \end{split}
   \end{equation*} 

   The above inequality was proved in Lemma~\ref{lemma:kotecky-luckhaus_A_1}.

\end{proof} 

\begin{remark}
   \label{rmk:dopo_stima_alto}
   From the definitions of $F''$ and $F''_{\infty}$ and  \eqref{eq:stima-alto-1} and \eqref{eq:stima-alto-2}, it follows
   \begin{equation}
      \label{eq:dopo_stima_alto_1}
     F''(u,A)  \lesssim |A| + \| \nabla u \|^p_{L^p(A)}
   \end{equation}
   and
   \begin{equation}
      \label{eq:dopo_stima_alto_2}
     F''_{\infty}(u,A)  \lesssim |A| + \| \nabla u \|^p_{L^p(A)}.
   \end{equation}
\end{remark}

\begin{lemma}[exponential tightness] 
   \label{lemma:exponential-tightness}
   Let $A$ be an open set and $K\geq 0$.   Denote by 
   \begin{equation*} 
      \begin{split}
         \Mcal_{K} := \insieme{\varphi:\ H(\varphi,A,\varepsilon) \geq K 
            \varepsilon^{-d} |A|}.
      \end{split}
   \end{equation*} 
   and
   \begin{equation*} 
      \begin{split}
         \Mcal^{\infty}_{K} := \insieme{\varphi:\ H_{\infty}(\varphi,A,\varepsilon) \geq K 
            \varepsilon^{-d} |A|}.
      \end{split}
   \end{equation*} 
   Then there exist constants $ D,K_{0},\varepsilon_{0} $ such that for every 
   $ K\geq K_{0},\ \varepsilon\leq \varepsilon_0  $ and $u\in L^{p}(A)$  it holds
   \begin{equation*} 
      \begin{split}
         \int_{\Mcal_{K}\cap \Vcal(u,A,\kappa)} \exp \big( -H(\varphi 
         ,A,\varepsilon )  \big) \leq \exp\Big(- \frac{1}{2}K\varepsilon^{-d} |A|+ 
         D |A |{\varepsilon^{-d}} -D\sum_{i-1}^{d}\sum_{x\in R_{\varepsilon }^{e_i}(A)}|\nabla_{e_i}\varphi_u  |^{p}\Big) 
      \end{split}
   \end{equation*} 
   and
   \begin{equation*} 
      \begin{split}
         \int_{\Mcal^{\infty}_{K}\cap \Vcal_{\infty}(u,A,\kappa)} \exp \big( -H_{\infty}(\varphi 
         ,A,\varepsilon )  \big) \leq \exp\Big(- \frac{1}{2}K\varepsilon^{-d} |A|+ 
         D |A |{\varepsilon^{-d}} -D\sum_{i-1}^{d}\sum_{x\in R_{\varepsilon }^{e_i}(A)}|\nabla_{e_i}\varphi_u  |^{p}\Big) .
      \end{split}
   \end{equation*} 
\end{lemma} 
\begin{proof} 
   For every $\varphi \in \Mcal_{K} $ it holds
   \begin{equation*} 
      \begin{split}
         H(\varphi,A,\varepsilon) \geq K/2 \varepsilon^{-d} + \frac{1}{2}H(\varphi,A,\varepsilon).
      \end{split}
   \end{equation*} 
   Hence, by using Lemma~\ref{lemma:stima-basso},  we have the desired result.

\end{proof}

We will now proceed to prove the hypothesis of Theorem~\ref{thm:integral_represenatation_sobolev}.

Even though in the next two lemmas a very similar reasoning is used, they cannot be derived one from the other.

\begin{figure}
   \label{fig:1}
   \begin{center}
      \includegraphics[scale=1]{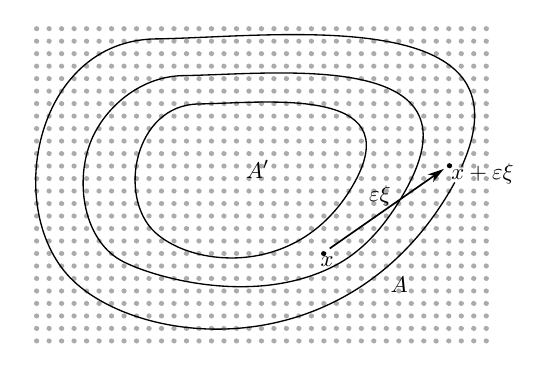}
   \end{center}
   \caption{}
\end{figure}

\begin{lemma}[regularity] 
   \label{lemma:regularity}
   Let $f_{\xi,\varepsilon} $ satisfy the usual hypothesis and $u\in W^{1,p}(\Omega)$. Then
   \begin{equation*} 
      \begin{split}
         \sup_{A'\Subset A} F''(u,A') = F''(u,A).
      \end{split}
   \end{equation*} 
\end{lemma}

\begin{proof} 

   Let us fix $A'\Subset  A$ and $N\in \N$ ({to be chosen later}).  Let $\delta =\dist(A',A^C)$, and let 
   $0< t_{1},\ldots,t_N\leq \delta$ such that $t_{ i+1 }-t_{i}> \frac{\delta }{2N}$. 
   Without loss of generality, we may assume that there exists no 
   $x\in A_{\varepsilon }$ such that $\dist(x,A^{C})=t_i$. 
   For every $i$, we define
   \begin{equation*} 
      \begin{split}
         A_{i}:=\insieme{x\in A_{\varepsilon }:\ \dist(x,A^{C})\geq t_i}
      \end{split}
   \end{equation*} 
   and
   \begin{equation*} 
      \begin{split}
         S^{\xi ,\varepsilon }_{i}:=\insieme{x\in  A_{i} :\ x+\varepsilon \xi \in A\setminus A_{i} }.
      \end{split}
   \end{equation*}

   With the above definitions it holds
   \begin{equation*} 
      \begin{split}
         R^{\xi }_{\varepsilon }(A) =  R^{\xi }_{\varepsilon }(A_i) \cup R^{\xi }_{\varepsilon }(A\setminus \bar{A_{i}}) \cup S^{\varepsilon ,\xi }_{i},
      \end{split}
   \end{equation*} 
   thus
   \begin{equation*} 
      \begin{split}
         H^{\xi }(\varphi ,A,\varepsilon ) \leq H^{\xi } (\varphi ,A\setminus 
         \bar{A_{i}},\varepsilon ) + H^\xi (\varphi ,A_{i},\varepsilon ) + 
         \sum_{x\in  S^{\xi ,\varepsilon }_{i}} f_{\xi ,\varepsilon }(\nabla \varphi (x)).
      \end{split}
   \end{equation*} 

   Hence, by using hypothesis \ref{cond:stima-alto} one has that,
   \begin{equation*} 
      \begin{split}
         H(\varphi ,A ,\varepsilon ) &\leq H(\varphi ,A_{i},\varepsilon ) + 
         H(\varphi ,A\setminus A_{i},\varepsilon ) + \sum_{\xi \in \Z^d}\sum_{x\in  S_{i}^{\xi 
               ,\varepsilon }} C_{\xi }(|\nabla_{\xi }\varphi (x) |^{p}+1).
      \end{split}
   \end{equation*} 

   Let us now estimate the last term in the previous inequality.  

   We separate the sum into two terms
   \begin{equation} 
      \label{eq:10121381574063}
      \begin{split}
         \sum_{\xi \in \Z^{d}} \sum_{x\in S^{\xi,\varepsilon  }_{i}} |\nabla_{\xi }\varphi (x)|^{p}= &
         \sum_{|\xi|\leq M} \sum_{x\in S^{\xi,\varepsilon  }_{i}} |\nabla_{\xi }\varphi (x)|^{p}+
         \sum_{ |\xi  |> M} \sum_{x\in S^{\xi,\varepsilon  }_{i}} |\nabla_{\xi }\varphi (x)|^{p},
      \end{split}
   \end{equation} 

   where $M\in \N$. 
   From hypothesis \ref{cond:stima-alto} and  by taking $M$ sufficiently large,  we may also assume  without loss of  generality that 
   \begin{equation*} 
      \begin{split}
         \sum_{|\xi|\geq M} C_{\xi} \leq \delta_1,
      \end{split}
   \end{equation*} 
   hence  using Lemma~\ref{lemma:zig-zag-sobolev},
   \begin{equation*} 
      \begin{split}
         \sum_{|\xi|\geq M} \sum_{x\in S^{\xi ,\varepsilon }_{i}} |\nabla_{\xi }\varphi (x) |^{p}\leq
         C \delta_1   \sum_{k=1}^{d} \sum_{x\in R^{e_{k} }_{\varepsilon }(A)} 
         |\nabla_{e_{k} }\varphi (x)  |^{p} \leq \tilde{C}\delta_{1} H(\varphi ,A,\varepsilon ),
      \end{split}
   \end{equation*} 
   where in the last inequality we have used hypothesis  \ref{cond:primi-vicini}.

   Let $|\xi|< M$.  If  $\varepsilon MN\leq 2\delta $, then 
   \begin{equation*} 
      \begin{split}
         S^{\xi,\varepsilon}_{i} \cap S^{\xi,\varepsilon}_{j}=\emptyset \qquad 
         \text{ whenever  } |i-j |\geq 2.
      \end{split}
   \end{equation*} 
   Without loss of generality we may assume the above condition given that we are interested in  $\varepsilon \to  0$.  

   Given that 
   \begin{equation*} 
      \begin{split}
         \sum_{i=1}^{N-2} \sum_{|\xi| < M}\sum_{x\in  S_{i}^{\varepsilon ,\xi }} 
         |\nabla_{\xi} \varphi (x) |^{p} \leq  2C H(\varphi ,A,\varepsilon ),
      \end{split}
   \end{equation*} 
   there exist $0<i\leq N-2$ such that 
   \begin{equation} 
      \label{eq:cagata1}
      \begin{split}
         \sum_{|\xi| <M} \sum_{x\in S_{i}^{\xi ,\varepsilon }}|\nabla_{\xi }\varphi |^{p}<\frac{2}{N-2}H(\varphi ,A,\varepsilon ).
      \end{split}
   \end{equation} 

   Let us denote by $\Ncal_{i}$ the set of all 
   $\varphi\in \Vcal(u,A,\kappa,\varepsilon )$ such that 
   \eqref{eq:cagata1} holds 
   for the first time, namely for every $j\leq i$ 
   \begin{equation} 
      \label{eq:cagata2}
      \begin{split}
         \sum_{|\xi| <M} \sum_{x\in S_{i}^{\xi ,\varepsilon }}|\nabla_{\xi }\varphi |^{p} <  \frac{2}{N-2}H(\varphi ,A,\varepsilon )
      \end{split}
   \end{equation} 
   On one side, one has that
   \begin{equation} 
      \label{eq:interp-stima-alto}
      \begin{split}
         \int_{\Vcal(u,A,\kappa ,\varepsilon )} \exp\left( -H(\varphi ,A,\varepsilon )   \right) \leq 
         \sum_{i=1}^{N}\int_{\Ncal_{i}}\exp\left( - H(\varphi 
            ,A_{i},\varepsilon) - H(\varphi ,A\setminus 
            \bar{A}_{i},\varepsilon )\right).
      \end{split}
   \end{equation} 
   On the other side, one has that
   \begin{equation} 
      \label{eq:interp-stima-baso}
      \begin{split}
         \int_{\Vcal(u,A,\kappa ,\varepsilon )} \exp\left( -H(\varphi ,A,\varepsilon )   \right) \geq 
         \sum_{i=1}^{N}\int_{\Ncal^{K}_{i}}\exp\left( -H(\varphi ,A,\varepsilon )  \right) ,
      \end{split}
   \end{equation} 
   where $\Ncal^{K}_{i}:= \Ncal_{i}\setminus \Mcal_{K}$.  By using  
   \eqref{eq:cagata2}, one has that for every $\varphi \in  \Ncal_{i}^{K}$ it holds
   \begin{equation*} 
      \begin{split}
         H(\varphi ,A,\varepsilon )\leq H(\varphi, A_{i}) + H(\varphi 
         ,A\setminus \bar{A}_{i}) + \frac{K|A |\varepsilon ^{-d}}{N-2},
      \end{split}
   \end{equation*} 
   and for every $\varphi$ it holds
   \begin{equation} 
      \label{eq:11191384880856}
      \begin{split}
         H(\varphi ,A,\varepsilon )\geq H(\varphi, A_i,\varepsilon ) + H(\varphi ,A\setminus \bar{A}_i,\varepsilon ) .
      \end{split}
   \end{equation} 

   Hence,
   \begin{equation*} 
      \begin{split}
         \int_{\Vcal(u,A,\kappa ,\varepsilon )} \exp\left( -H(\varphi ,A,\varepsilon )     \right) \geq  
         \sum_{i=1}^{N}\int_{\Ncal^{K}_{i}}\exp\left( - H(\varphi, A_{i},\varepsilon) - H(\varphi 
            ,A\setminus \bar{A}_{i},\varepsilon) - \frac{K |A|\varepsilon^{-d}}{N-2} \right).
      \end{split}
   \end{equation*} 
   By using Lemma~\ref{lemma:exponential-tightness}, \ie the fact that there exist 
   $K_{0},\ \varepsilon_{0}$  and $D$ such that for every $K>K_{0}$ and 
   $\varepsilon \leq \varepsilon_{0}$ one has that
   \begin{equation*} 
      \begin{split}
         \int_{\Mcal_{K}\cap \Vcal(u,A,\kappa ,\varepsilon )} \exp \big( 
         -H(\varphi ,A,\varepsilon )  \big) \leq \exp\Big(- 
         \frac{1}{2}K\varepsilon^{-d}|A|+ D\varepsilon^{-d}|A|\Big),
      \end{split}
   \end{equation*} 
   and by using \eqref{eq:interp-stima-alto}, one has that 
   \eqref{eq:interp-stima-baso} can be further  estimated as
   \begin{equation*} 
      \begin{split}
         \exp\Big(-\frac{K\varepsilon^{-d}}{N-2} -& \frac{1}{2}K\varepsilon ^{-d}|A | + D \varepsilon 
         ^{-d}|A|\Big) + \int_{\Vcal(u,A,\kappa ,\varepsilon )}\exp\left( - H (\varphi 
            ,A,\varepsilon )\right) \\ \geq & \exp\left(-\frac{K\varepsilon^{-d}}{N-2}\right) \sum_{i=1}^{N}\int_{\Ncal_{i}} \exp\left(   
            -H(\varphi ,A_{i},\varepsilon )- H(\varphi ,A\setminus \bar{A}_{i})  \right).
      \end{split} 
   \end{equation*} 

   We also notice that by using \eqref{eq:11191384880856} one has that
   \begin{equation*} 
      \begin{split}
         \sum_{i=1}^{N}\int_{\Ncal_{i}}\exp\big( -H(\varphi ,A_i,\varepsilon )- H(\varphi ,A\setminus \bar{A}_{i},\varepsilon)\big) \geq \int 
         _{\Vcal(u,A,\varepsilon )}  \exp\left( -H(\varphi ,A,\varepsilon ) \right),
      \end{split}
   \end{equation*} 
   thus there exists $1 \leq i_{0}\leq N$ such that 
   \begin{equation} 
      \label{eq:11191384882796}
      \begin{split}
         \int_{\Ncal_{i_{0}}}\exp\big(-  H(\varphi ,A_{i_0},\varepsilon )- H(\varphi ,A\setminus \bar{A}_{i_0},\varepsilon)\big) \geq \frac{1}{N} \int 
         _{\Vcal(u,A,\varepsilon )} \exp\left( -H(\varphi ,A,\varepsilon )\right).
      \end{split}
   \end{equation}

   Without loss of generality, we may assume that $i_0=1$.   
   Hence, combining \eqref{eq:11191384882796} with the previous estimates we have that 
   \begin{equation} 
      \label{eq:mstr1}
      \begin{split}
         \exp\Big(-\frac{K|A| \varepsilon^{-d}}{N-2} &- \frac{1}{2}K\varepsilon ^{-d}|A | + D \varepsilon 
         ^{-d}|A|\Big)  +   \int_{\Vcal(u,A,\kappa ,\varepsilon )}\exp\left( - H (\varphi 
            ,A,\varepsilon )\right) \\ &\geq \frac{1}{N}\exp\left(-\frac{K|A | \varepsilon^{-d}}{N-2}\right) 
         \int_{\Ncal_{1}} \exp\left(   
            -H(\varphi ,A_{1},\varepsilon )- H(\varphi ,A\setminus 
            \bar{A}_{1})  
         \right). 
      \end{split}
   \end{equation}

   We notice that the variables $H(\varphi ,A_1,\varepsilon  ) $ and  
   $H(\varphi ,A\setminus \bar{A}_1,\varepsilon  ) $ are 
   independent, thus by using the Fubini theorem one has that
   \begin{equation} 
      \label{eq:mstr2}
      \begin{split}
         \int_{\Vcal(u,A,\kappa ,\varepsilon )} \exp\left(   
            -H(\varphi ,A_{1},\varepsilon )- H(\varphi ,A\setminus 
            \bar{A}_{1})\right)\geq  \int_{\Vcal(u,A_{1},\kappa ,\varepsilon )} \exp\left(   
            -H(\varphi ,A_{1},\varepsilon )\right)   \\ \times  \int_{\Vcal(u,A\setminus \bar{A}_{1},\kappa ,\varepsilon )} \exp\left(   
            - H(\varphi ,A\setminus 
            \bar{A}_{1})\right),
      \end{split}
   \end{equation} 
   where in the previous inequality we have also used that
   \begin{equation*} 
      \begin{split}
         \Vcal(u,A\setminus \bar{A}_{1},\kappa ,\varepsilon ) \cap\Vcal(u,{A}_{1},\kappa ,\varepsilon )\subset \Vcal(u,A,\kappa ,\varepsilon ).
      \end{split}
   \end{equation*} 

   From \eqref{eq:mstr1} and \eqref{eq:mstr2}, we have that
   \begin{equation*}
      F(u,A,\kappa,\varepsilon) \leq \varepsilon^d\log(N)  + \frac{K}{N-2} + F(u,A_1,\kappa,\varepsilon) + F(u,A\setminus \bar{A}_1,\kappa,\varepsilon).
   \end{equation*}

   Finally, to conclude it is enough to pass to the limit in $\varepsilon$, then in $N$ and then in $\kappa$, and  use the monotonicity of the map $A\mapsto F''(u,A)$(see Remark~\ref{rmk:almost-monotonicity} ) and Lemma~\ref{lemma:stima-alto} to estimate the term  $F(u,A\setminus \bar{A}_{1},\kappa ,\varepsilon)$. 
\end{proof}

\begin{lemma} 
   \label{lemma:cambio-dati-bordo}
   Let $u\in W^{1,p}(\R^d)$,  and let  $A$ be an open set with piecewise regular regular boundary, suppose that $\partial A$ has finite length,  and let $u\in W^{1,p}(\R^{d})$. Then the followings hold
   \begin{equation*} 
      \begin{split}
         F'(u,A)=F'_{\infty }(u,A)\qquad\text{and}\qquad  F''(u,A)=F''_{\infty }(u,A).
      \end{split}
   \end{equation*} 
\end{lemma}

\begin{proof}



   Given that $H_{\infty}(\varphi,A,\varepsilon) \geq H(\varphi,A,\varepsilon)$, it is not difficult to notice that 
   \begin{equation*} 
      \begin{split}
         F'_{\infty}(u,A,\kappa,\varepsilon) \geq F'(u,A,\kappa,\varepsilon)\quad \text{and}\quad F'_{\infty}(u,A,\kappa,\varepsilon) \geq F'(u,A,\kappa,\varepsilon).
      \end{split}
   \end{equation*} 

   The rest of the proof will consist in proving the opposite inequalities. 

   Let $A' \subset\!\subset A$ and a family of sets $(A_i)_{i=1}^n$ so that   
   \begin{equation*} 
      \begin{split}
         A_{i} := \insieme{x\in A_{\varepsilon}:\ \dist(x,A^C) > t_i} \quad \text{and}\quad S_{i}^{\xi,\varepsilon}:= \insieme{x\in A_i:\ x + \varepsilon \xi \not\in A_i}
      \end{split}
   \end{equation*} 
   here $t_i = i \frac{\dist(A',A^C)}{n}$.  From  hypothesis (C2) and  Lemma~\ref{lemma:zig-zag-sobolev}, one has that for every $\delta_0$, there exists $N\in \N$ such that 
   \begin{equation*} 
      \begin{split}
   \sum_{|\xi | > N}\sum_{x\in A_{\varepsilon}} f_{\xi,\varepsilon}(x,\nabla_{\xi}\varphi) \leq \sum_{|\xi | > N}\sum_{x\in A_{\varepsilon}} C_{\xi} (|\nabla \varphi(x) |^p + 1) \leq \delta_{0}\sum_{x\in A_\varepsilon} \sum_{j=1}^{d}  (|\nabla_j \varphi(x) |^p + 1). 
      \end{split}
   \end{equation*} 
   Thus, it holds
   \begin{equation*} 
      \begin{split}
         \sum_{\xi\in \Z^d} \sum_{x\in A_{\varepsilon}} f_{\xi,\varepsilon}(x,\nabla \varphi(x))   \leq  \sum_{|\xi| <N} \sum _{x\in A_{\varepsilon}} f_{\xi,\varepsilon}(x,\nabla \varphi(x))  + \delta_{0} \sum_{x\in A_{\varepsilon}}\sum_{j=1}^{d}(|\nabla_j \varphi(x) |^{p} + 1).
      \end{split}
   \end{equation*} 
   The right hand side can be rewritten by
   \begin{equation*} 
      \begin{split}
         \sum_{|\xi|<N}  \sum _{x\in A_i} f_{\xi,\varepsilon}(x,\nabla_{\xi} 
	 \varphi(x)) & + \sum_{|\xi| < N}\sum_{x\in S_{i}^{\xi,\varepsilon}} 
	 f_{\xi,\varepsilon}(x,\nabla_{\xi} \varphi(x))+ 
            \sum_{|\xi| <N} \sum _{x\in A_\varepsilon\setminus A_i} 
	    f_{\xi,\varepsilon}(x,\nabla_\xi \varphi(x)) \\ &+ \delta_{0} 
	    \sum_{x\in A_{\varepsilon}} (|\nabla \varphi(x) |^{p}+1 ).
      \end{split}
    \end{equation*} 
    Thus, there exists a constant $C_{N}$ (depending on $N$) such that
   \begin{equation*} 
      \begin{split}
         \sum_{i=1}^n\sum_{|\xi|<N} \sum_{x\in S_{i}^{\xi,\varepsilon}} 
f_{\xi,\varepsilon}(x,\nabla \varphi(x)) \leq C_N H(\varphi,A,\varepsilon).
      \end{split}
   \end{equation*} 
    From the above one has that  there exists $1\leq i \leq n$
    \begin{equation} 
       \label{eq:proofdatiBordo1}
       \begin{split}
          \sum_{x\in S_{i}^{\xi,\varepsilon}}f_{\xi,\varepsilon}(x,\nabla\varphi(x))  \leq \frac{C_N}{n} H(\varphi,A,\varepsilon). 
       \end{split}
    \end{equation} 

    Thus, the range of interractions $N$ will be fixed, and  take $n$ large.  Combining the above we have 
    \begin{equation*}
	    H_{\infty} (\varphi,A,\varepsilon) \leq H(\varphi,A_i,\varepsilon) + \frac{C_N}{n} H_{\infty}(\varphi,A,\varepsilon) + H_{\infty}(\varphi,A\setminus A_i,\varepsilon),
    \end{equation*}
    where in the we have estimated
    \begin{equation*}
       \sum_{|\xi| <N} \sum _{x\in A_\varepsilon\setminus A_i} 
	    f_{\xi,\varepsilon}(x,\nabla_\xi \varphi(x))\leq H_{\infty}(\varphi,A\setminus A_{i},\varepsilon).
    \end{equation*}

    As in the proof of Lemma~\ref{lemma:regularity}, by 
    using Lemma~\ref{lemma:exponential-tightness} 
    one can show that there exists ${i\in \{1,\ldots,n\}}$ and $K$(depending 
    eventually on $N,\delta_0,p$), such that
         $\limsup_{\varepsilon\to 0} F_{\infty}(0,A,\kappa,\varepsilon) $ can 
	 be bounded from above by
    \begin{equation*}
	    \limsup_{\varepsilon\downarrow 0} -\varepsilon^d\log
	 \int_{\Vcal(0,A,\kappa,\varepsilon)} \exp\left( -H(\varphi,A_i,\varepsilon)- 
		 H_{\infty}(\varphi,A\setminus A_i,\varepsilon) - K|A|/n \varepsilon ^{-d} \right).
    \end{equation*}
    
    Thus by passing to the limit in $\varepsilon$ and $\kappa$
    \begin{equation*}
       F''_{\infty}(u,A) \leq  \sup_{i} F''(u,A_{i}) + F''_{\infty}(u,A\setminus A').
    \end{equation*}
    In order to conclude, we use Lemma~\ref{lemma:regularity} and \eqref{eq:dopo_stima_alto_2}.

\end{proof}

\begin{lemma}[subadditivity] 
   \label{lemma:subadditivity}
   Let $A',A,B',B \subset \Omega $ be open sets such that $A'\Subset A $ and 
   such that $B' \Subset B $.  Then for every $u\in W^{1,p}$ one has that
   \begin{equation*} 
      \begin{split}
         F''(u,A'\cup B') \leq F''(u,A) + F''(u,B).
      \end{split}
   \end{equation*} 
\end{lemma} 

\begin{proof} 

   The proof is very similar to Lemma~\ref{lemma:regularity}. We will show the main steps needed to prove the desired claim.

   As in the proof of Lemma~\ref{lemma:regularity}, fix and $N\in \N$ ({to be chosen later}).  
   Let $\delta =\dist(A',A^C)$, and let $0< t_{1},\ldots,t_N\leq \delta$ such that $t_{ i+1 }-t_{i}> \frac{\delta }{2N}$. 

   Without loss of generality, we may assume that there exists no $x\in A_{\varepsilon }$ such that $\dist(x,A^{C})=t_i$. 
   For every $i$, we define
   \begin{equation*} 
      \begin{split}
         A_{i}:=\insieme{x\in A_{\varepsilon } :\ \dist(x,A^{C})\geq t_i}
      \end{split}
   \end{equation*} 
   and
   \begin{equation*} 
      \begin{split}
         S^{\xi ,\varepsilon }_{i}:=\insieme{x\in  A_{i} :\ x+\varepsilon \xi \in A\setminus A_{i} }.
      \end{split}
   \end{equation*} 

   With the above definitions it holds
   \begin{equation*} 
      \begin{split}
         R^{\xi }_{\varepsilon }(A' \cup B') =  R^{\xi }_{\varepsilon }(A') \cup R^{\xi }_{\varepsilon }(B'\setminus \bar{A_i}) \cup S^{\varepsilon ,\xi }_{i},
      \end{split}
   \end{equation*} 
   thus
   \begin{equation*} 
      \begin{split}
         H^{\xi }(\varphi ,A' \cup B',\varepsilon ) \leq H^{\xi } (\varphi ,B'\setminus 
         \bar{A_{i}},\varepsilon ) + H^\xi (\varphi ,A_{i},\varepsilon ) + 
         \sum_{x\in  S^{\xi ,\varepsilon }_{i}} f_{\xi ,\varepsilon }(\nabla \varphi (x)).
      \end{split}
   \end{equation*} 

   By using hypothesis \ref{cond:stima-alto}, and Lemma~\ref{lemma:zig-zag-sobolev} one has that,
   \begin{equation} 
      \label{eq:mstr3}
      \begin{split}
         H(\varphi ,A' \cup B' ,\varepsilon ) &\leq H(\varphi ,A_{i},\varepsilon ) + 
         H(\varphi ,B'\setminus A_{i},\varepsilon ) + \sum_{\xi \in \Z^d}\sum_{x\in  S_{i}^{\xi ,\varepsilon }} C_{\xi }(|\nabla_{\xi }\varphi (x) |^{p}+1).
      \end{split}
   \end{equation} 

   By using the tightness Lemma as in the proof of Lemma~\ref{lemma:regularity}, we can restrict ourselves to 
   restricting to $\Mcal_K$, namely
   \begin{equation*}
      \Mcal_K : = \insieme{ \varphi: H(\varphi, A\cup B, \varepsilon) \leq K|A\cup B |\varepsilon^{-d} }.
   \end{equation*}

   We then proceed to write the last term in \eqref{eq:mstr3} as the sum of two terms
   \begin{equation}
      \label{eq:mstr4}
      \sum_{\xi \in \Z^d}\sum_{x\in  S_{i}^{\xi }} C_{\xi }(|\nabla_{\xi }\varphi (x) |^{p}+1) \leq 
      \sum_{|\xi| > M}\sum_{x\in  S_{i}^{\xi }} C_{\xi }(|\nabla_{\xi }\varphi (x) |^{p}+1) +
      \sum_{|\xi| \leq M}\sum_{x\in  S_{i}^{\xi }} C_{\xi }(|\nabla_{\xi }\varphi (x) |^{p}+1) 
   \end{equation}
   The first term in the r.h.s.  of \eqref{eq:mstr4}, can be dealt in the same way as in Lemma~\ref{lemma:regularity}.
   Similarly to \eqref{eq:cagata1}, we have that that there exist $0<i\leq N-2$ such that 
   \begin{equation*} 
      \begin{split}
         \sum_{|\xi| <M} \sum_{x\in S_{i}^{\xi ,\varepsilon }}|\nabla_{\xi }\varphi |^{p}<\frac{2}{N-2}H(\varphi ,A \cup B,\varepsilon ).
      \end{split}
   \end{equation*} 
   After this step the proof continues in the same manner as the proof of Lemma~\ref{lemma:regularity}.
\end{proof}

\begin{lemma}[locality] 
   \label{lemma:locality-sobolev}
   Let $u,v \in W^{1,p}(\Omega)$ such that $u\equiv v $ in $A$.   Then 
   \begin{equation} 
      \label{eq:lemma-locality}
      \begin{split}
         F'(u,A)=F'(v,A)\qquad \text{and}\qquad F''(u,A)=F''(v,A)
      \end{split}
   \end{equation} 
\end{lemma} 

\begin{proof} 
   The statement follows from the definitions.  

\end{proof} 

\begin{proof}[Proof of {Theorem~\ref{thm:main-representation-sobolev-statmech}}] 
   \hfill

   Let us suppose initially that there exists a sequence for which 
   $F(\cdot,\cdot)=F'(\cdot,\cdot)= F''(\cdot,\cdot)$.  Then to conclude it is 
   enough to notice that $F $ satisfies the conditions of 
   Theorem~\ref{thm:integral_represenatation_sobolev}. 
   Indeed, in the previous Lemmas we prove that all the conditions (i)-(v) of Theorem~\ref{thm:integral_represenatation_sobolev} 
   hold.  
\end{proof}

\begin{corollary} 
   Because of Lemma~\ref{lemma:cambio-dati-bordo}, 
   the same statement holds true for $F_{\infty} $. This in particular implies 
   that for the sequence $\insieme{\varepsilon_{n_{k}}}$ in 
   Theorem~\ref{thm:main-representation-sobolev-statmech}  there holds a large 
   deviation principle with rate functional 
   \begin{equation}
      I(v)=\int_{\Omega} W(x,\nabla v) \dx - \min_{\bar v\in W^{1,p}_0(\Omega)+u}\,  
      \int_{\Omega} W(x,\nabla \bar v(x))dx. 
   \end{equation} 
\end{corollary}

\subsection{Homogenisation} 
\label{sub:Homogenisation}

In this section we will show that if the functions $f_{\xi ,\varepsilon }$ are 
obtained by rescaling by $\varepsilon $ in the space variable, then a LDP 
result holds true. This models the case when the arrangement of the ``material 
points'' presents a periodic feature, namely:
\begin{enumerate}[label=({H}\arabic*)] 
   \item\label{hypothesis:homog-periodicity} periodicity:
      \begin{equation*} 
         \begin{split}
            f_{\xi,\varepsilon}(x,t) = f^{\xi}\Big( \frac{x}{\varepsilon}, t\Big)
         \end{split}
      \end{equation*} 
      where the functions $f^{\xi }$ are such that  $f^{\xi }(x+M e_{i},t)=f^{\xi }(x,t)$.  

   \item\label{hypothesis:homog-lowerbound-nearest-neighbours} lower bound on the nearest neighbours:
      \begin{equation*} 
         \begin{split}
            f^{e_i} (x,t) \geq c_1 (|t|^{p} -1)
         \end{split}
      \end{equation*} 
   \item\label{hypothesis:homog-upper-bound} upper bound
      \begin{equation*} 
         \begin{split}
            f^{\xi}(x,t) \leq C_{\xi} (|t |^{p} +1)
         \end{split}
      \end{equation*} 

\end{enumerate}

The next homogenization result follows from Theorem~\ref{thm:main-representation-sobolev-statmech} and adapting homogenization arguments to our setting. 
For a similar result on the discrete see also~\cite{MR2083851}.

\begin{theorem} 
   \label{thm:sobolev_homog_main}
   Let the functions $f_{\xi ,\varepsilon}^{\xi} $ satisfy the above 
   conditions. Then there exists a function $f_{\hom}$ such that for every  
   $A \subset \Omega$ open set it holds
   \begin{equation} 
      \label{eq:teo-homog_main_eq1}
      \begin{split}
         F(u,A)= \begin{cases} 
            \int_{A} f_{\homo}(\nabla u)& \text{ if }u\in W^{1,p}(\Omega;\R^d)\\
            +\infty& \text{otherwise, }
         \end{cases} 
      \end{split}
   \end{equation} 
   where
   \begin{equation} 
      \label{eq:definizione-f-ldp}
      \begin{split}
         f_\homo(M):=\frac{1}{|A |} \lim_{\varepsilon \downarrow 0}F'(Mx,A,\kappa ,\varepsilon ).
      \end{split}
   \end{equation} 

\end{theorem}

\begin{proof} 
   The proof is an adaptation 

   Let $(\varepsilon _n)$ be a sequence of positive numbers converging
   to $0$. From Proposition~\ref{thm:main-representation-sobolev-statmech} 
   we can extract a subsequence (that we do not relabel for simplicity) such that
   \begin{equation*} 
      \begin{split}
         F'_{\{\varepsilon_{n}\}}(u,A)= F''_{\{ \varepsilon_{n} \}}(u,A) = 
         \int_A f_{\{ \varepsilon _{n}\}}(x,\nabla u) \dx.   
      \end{split}
   \end{equation*} 

   The theorem is proved if we show that $f$ does not depend on the
   space variable $x$ and on the chosen sequence $\varepsilon _{n}$. 
   To prove the first claim, by Theorem~\ref{thm:integral_represenatation_sobolev},  it suffices to show that, if one denotes by 
   \begin{equation*} 
      \begin{split}
         F(u,A)=\int_A f(x,\nabla u)\dx,
      \end{split}
   \end{equation*} 
   then 
   \begin{equation*} 
      \begin{split}
         F(Mx,B(y,\rho))=F(Mx,B(z,\rho))
      \end{split}
   \end{equation*} 
   for all $M\in {\R}^{d\times m}$, $y,z\in\Omega $ and $\rho>0$ such
   that $B(y,\rho)\cup B(z,\rho)\subset\Omega $. We will prove that
   \begin{equation*} 
      \begin{split}
         F(Mx,B(y,\rho))\leq F(Mx,B(z,\rho)).
      \end{split}
   \end{equation*} 
   The proof of the opposite inequality is analogous.

   Let $x,y\in \R^{d}$ and let 
   $x_\varepsilon =\arg \min(\dist(y,x + (\varepsilon  M )\Z^{d})) \big]
   $. Then $x_{\varepsilon } \to y$ as $\varepsilon \downarrow 0$.  
   From the periodicity hypothesis, one has that 
   \begin{equation*} 
      \begin{split}
         F\big(M,B(x,\rho,\kappa ,\varepsilon )\big) = F\big(M,B(x_{\varepsilon },\rho,\kappa 
         ,\varepsilon )\big) \leq  F\big(M,B(y,\rho +\delta ,\kappa 
         ,\varepsilon )\big)
      \end{split}
   \end{equation*} 
   where in the last inequality we have used the monotonicity with respect to the 
   inclusion relation of $A\mapsto F(u,A,\kappa ,\varepsilon )$ and $\delta $ is 
   such that $|y-x_{\varepsilon }|\leq \delta$.

   Let us now turn to the independence on the sequence on the chosen sequence. 
   Let us initially notice that because of the LDP, whenever $u=Mx$ where $M$ is 
   a linear map it holds 
   \begin{equation} 
      \label{eq:conclusione-utile-ldp}
      \begin{split}
         F'(u,A,\kappa)=F'(u,A)\qquad\text{and}\qquad F''(u,A,\kappa)=F''(u,A).
      \end{split}
   \end{equation} 

   Because of Theorem~\ref{thm:integral_represenatation_sobolev}, it is enough to  show that for every linear map $M$ the following limit exists and 
   \begin{equation*} 
      \begin{split}
         \frac{1}{|A |} \lim_{\varepsilon \downarrow 0}F'(Mx,A,\kappa ,\varepsilon)
      \end{split}
   \end{equation*} 

   The existence of the above limit(and its independence on $\kappa$) follows easily by  the standard methods with the help of an approximative subadditivity. 
   A simple proof can be found in \cite[Proposition~1.2]{LK-comm-math-phys}.
\end{proof}

\section{$\sbv$ Representation Theorem} 
\label{sec:sbv-representation-theorems}

In this section we extend the results of the previous section to more general local interactions, where the problem relaxes naturally in $\sbv$. 
The strategy will be very similar to the one used in Section~\ref{sec:sobolev-representation}.
One key argument in the previous section  was  Theorem~\ref{thm:integral_represenatation_sobolev}. We will replace this result from the literature with a similar  representation theorem for $\sbv$ functions. 
%

\subsection{A \emph{very} short introduction to $\sbv$} 

Before going into the details of our main Theorem of this section, let us 
define the functional spaces $\bv$ and $\sbv$.  For a general introduction on these spaces see \cite{AFPBV}.
However, please notice that the definitions given in this section differ slightly from the ones in \cite{AFPBV}. 
More precisely, in the following,  we additionally impose  the finiteness of $(n-1)$-Hausdorff measure of the jump set. 
This technical modification is done in order to be able to use the  general representation theorems, namely Theorem~\ref{thm:integral_represenatation_sbv-sbv}

Let $\Omega $ be an open set.  
We say that $u\in L^{1}(\Omega )$ belongs to $\bv(\Omega )$, if there exists a vector measure
$Du = (D_{1}u,\ldots, D_{n}u)$ with finite total variation in $\Omega $, 
such that 
\begin{equation*} 
   \begin{split}
      \int_{\Omega } u \partial _{i} \varphi  \dx = - \int \varphi  \d D_{i}u 
      \qquad \forall \varphi \in C^{1}_{0}(\Omega)
   \end{split}
\end{equation*} 

Let $Du = D^{a }u +D^{s} u $ be the Radon-Nikodym decomposition of $Du$ in 
absolutely continuous and singular part with respect to the $\leb^n$ and let 
$\nabla u$ be the density of $D^{a}u$.   It can be seen that $u$ is approximately 
differentiable at $x$ and the approximate differential equals to $\nabla u(x)$, 
\ie
\begin{equation*} 
   \begin{split}
      \lim_{\rho \downarrow 0} \rho ^{-n}\int_{B_{\rho }(x)}   \frac{|u(y) - 
         u(x) - \scal{\nabla u}{y-x} |}{|y-x |}\dy =0
   \end{split}
\end{equation*} 
for $\leb^{n}$-\ae   $x\in \Omega $.  

For the singular part, it is useful to introduce the upper and lower 
approximate limits $u_{+},u_{-}$, defined by
\begin{equation*} 
   \begin{split}
      u_{-}(x) = \inf \insieme{t\in [-\infty ,+\infty ]:\ \{ x\in \Omega :\ u(x) >t\} \text{ has density $0$ at } x}\\
      u_{+}(x) = \sup \insieme{t\in [-\infty ,+\infty ]:\ \{ x\in \Omega :\ u(x) <t\} \text{ has density $0$ at } x}.
   \end{split}
\end{equation*} 
It is well-known that $u_{+}(x)\in \R$ for $\hausd^{d-1}$-\ae  $x\in \Omega $.  
The jump set $S_u$ is defined by
\begin{equation*} 
   \begin{split}
      S_u:= \insieme{x\in : \ u_{-}(x)<u_{+}(x)}.
   \end{split}
\end{equation*} 
We define the jump part $Ju$ of the derivative  as the restriction of $D^{s}u$ 
to the jump set $S_{u}$.  
We also recall that there exists a Borel map $\nu _{u}:S_{u}\to S^{d-1}$ such that
\begin{equation*} 
   \begin{split}
      \nu_{E_t}(x) =\nu _{u} \qquad\text{for }\hausd^{d-1}\text{\ae 
         $x\in \partial^{*}E_t\cap S_{u}$ }
   \end{split}
\end{equation*} 
for any $t$ such that $E_{t}:=\insieme{x: u>t}$.  
\begin{proposition} 
   \label{prop:jump-part-sbv}
   Let $u\in \bv(\Omega)$.   Then, the jump part of the derivative is 
   absolutely continuous with respect to $\hausd^{d-1}$ and 
   \begin{equation*} 
      \begin{split}
         Ju = (u_{+} - u_{-})\nu_{u} \hausd^{d-1}\res S_{u}
      \end{split}
   \end{equation*} 
\end{proposition}

Finally, we define the space  $\sbv_{p}(\Omega )$  as the set of functions 
$u\in \bv(\Omega )$ such that $\nabla u\in L^{p}(\Omega )$ $D^{s}u = Ju$ 
and 
\begin{equation} 
   \label{eq:cond-jump-extra}
   \begin{split}
      \hausd^{d-1}(S_u)<+\infty.
   \end{split}
\end{equation} 

Note that in \cite{AFPBV} the condition \eqref{eq:cond-jump-extra} is not imposed.

\subsection{Preliminary results}

Let us now recall some well-known results, which will be useful in the sequel. 

\begin{theorem}[{\cite[Theorem~4.7]{AFPBV}}] 
   \label{amb_theorem_4.7-sbv}
   Let $\varphi:[0,+\infty)\to [0,+\infty]$, $\theta: [0,+\infty)\to 
   [0,+\infty]$ be a lower semicontinuous function 
   increasing functions and assume that 
   \begin{equation}
      \label{afp_4.3-sbv}
      \begin{split}
         \lim_{t\to +\infty } 
         \frac{\varphi(t)}{t}= +\infty \qquad 
         \text{and} \qquad \lim_{t\to 0 }\frac{\theta(t)}{t} =+\infty
      \end{split}
   \end{equation}
   Let $\Omega \subset \R^n $ be an open and bounded and let 
   $(u_h)\subset \sbv(\Omega)$ such that  
   \begin{equation}
      \label{afp_4.4-sbv}
      \begin{split}
         \sup \insieme{\int_{\Omega} 
            \varphi(|\nabla u_h|)+\int_{J_{ u_h }}\theta (|u^{+}_h -u^{-}_h |) \d \hausd^{n-1}} <+\infty.  
      \end{split}
   \end{equation}
   If $(u_h)$ weakly* converges in $\bv(\Omega)$, then $u\in \sbv(\Omega)$, the 
   approximate gradients $\nabla u_h$ weakly converge to 
   $\nabla u \in (L^{1}(\Omega))^N$.   $D_j u_h$ weakly* converge to 
   $D_j u\in \Omega$  and 
   \begin{equation*}
      \begin{split}
         \int_{\Omega} \varphi(|\nabla u|)\dx \leq \liminf_{h\to +\infty} 
         \int_{\Omega}  \varphi (|\nabla u_h|) \dx \qquad \text{if } \varphi \text{\ is convex}\\
         \int _{J_u} \theta(|u^+ -u^-|) \d \hausd^{n-1}\leq \liminf _{h\to 
            +\infty} \int_{J_{u_h}} \theta(|u_h^+ -u_h^-|) \d\hausd^{n-1} 
      \end{split}
   \end{equation*}
   if $\theta$ is concave.  

\end{theorem}

\begin{theorem}[Compactness $\sbv$ {\cite[Theorem~4.8]{AFPBV}}] 
   \label{afp_theorem_4.8-sbv}
   Let $\varphi,\ \theta $ as in Theorem~\ref{amb_theorem_4.7-sbv}. Let $(u_h)$ in 
   $\sbv(\Omega)$ satisfy \eqref{afp_4.4-sbv} and assume in addition that 
   $\|u_h\|_{\infty}$ is uniformly bounded in $h$.  Then there exists a 
   subsection $(u_{h_k})$ weakly* converging in $\bv(\Omega)$ to $u\in \sbv(\Omega)$.  
\end{theorem}

\ 

We now give the a set of condition which give a representation formula similar 
to the one of  Theorem~\ref{thm:integral_represenatation_sobolev}.

Let
\begin{equation*}
   \begin{split}
      \mathcal F : \sbv_p(\Omega,\R^d)\times \mathcal A(\Omega) \to [0,+\infty]
   \end{split}
\end{equation*}
such that the followings hold:
\begin{enumerate}[label=(H\arabic*)]
   \item	$\mathcal F(u,\cdot)$ is the restriction to $\mathcal A(\Omega)$ 
      of a Radon measure,
   \item $\mathcal F (u,A)=\mathcal F(v,A)$ whenever $u=v\ \leb^n$ \ae on 
      $A\in \mathcal A(\Omega)$,  
   \item $\mathcal F(\cdot,A)$ is $L^1$ l.s.c.,
   \item there exists a constant $C$ such that 
      \begin{equation}
         \label{eq:fonseca-condition-iv}
         \begin{split}
            \frac{1}{C} &\bigg( \int_{A} |\nabla u|^p \dx  + 
            \int_{S(u)\cap A}\big( 1+ |u^+ -u^-| \big)\d\hausd^{n-1} 
            \bigg)\\& \leq \mathcal F(u,A)\\& \leq C\bigg( \int_{A} |\nabla u|^p \dx + 
            \int_{S(u)\cap A}\big( 1+ |u^+ -u^-| \big)\d\hausd^{n-1} 
            \bigg).
         \end{split}
      \end{equation}
\end{enumerate}
Here, $\Omega$ is an open bounded set of $\R^n$. As before,  $\Acal(\Omega)$ is the class of all open subsets of $\Omega$ and  $\sbvp(\Omega)$ is the space of functions $u\in \sbv(\Omega)$ such that ${\nabla}u\in L^p(\Omega)$ and $\hausd^{n-1}\big(J_u\big)< +\infty$.  
For every $u\in \sbvp(\Omega)$ and $A\in \Acal(\Omega)$ define 
\begin{equation*}
   \begin{split}
      m(u;A):=\inf \insieme{\Fcal(u;A):\ w\in \sbvp(\Omega) \ \text{such that 
            $w=u$ in a neighbourhood of $\partial A$}}
   \end{split}
\end{equation*}

The role of Theorem~\ref{thm:integral_represenatation_sobolev}, will be played by the following result, whose proof can be founded in~\cite{MR1941478}.

\begin{theorem} 
   \label{thm:integral_represenatation_sbv-sbv}
   Under hypotheses (H1)-(H4), for every $u\in \sbvp(\Omega)$ and $A\in \Acal(\Omega)$ there exists a function $W_{1}$ and $W_{2}$ such that  $W_{1}$ is quasi-convex, $W_{2}$ is $\bv$-elliptic and such that 
   \begin{equation*}
      \begin{split}
         \Fcal(u,A):=\int _{A} W_{1}(x,u,\nabla u)\dx + \int_{A\cap S_u} W_{2}(x,u^+,u^-,\nu_u)\d\hausd^{n-1}.
      \end{split}
   \end{equation*}
   Moreover, the functions $W_{1}$ and $W_{2}$ can be computed via
   \begin{equation*}
      \begin{split}
         W_1(x_0,u_0,):=\limsup _{\varepsilon\to 0^+} \frac{ \boldsymbol{m}\big(u_0+\xi(\cdot - x_0),Q(x_0,\varepsilon)\big)}{\varepsilon^n}
      \end{split}
   \end{equation*}
   \begin{equation*}
      \begin{split}
         W_2(x_0,a,b,\nu):=\limsup_{\varepsilon\to 0^+} \frac{ \boldsymbol{m}\big(u_{x_0,a,b,\nu},Q_{\nu}(x_0,\varepsilon) \big)}{\varepsilon^{ n-1 }}
      \end{split}
   \end{equation*}
   for all $x_0\in \Omega,$ $ u_0$,  $a,b\in \R^d$, $\xi\in \R^d$, 
   $\nu\in S^{n-1}$ and where 
   \begin{equation*}
      \begin{split}
         u_{x_0,a,b,\nu}(x):= \begin{cases}
            a &\text{if }(x-x_0)\cdot\nu >0,\\
            b &\text{if }(x-x_0)\cdot\nu \leq 0.
         \end{cases}
      \end{split}
   \end{equation*}
   As $u_{x_0,b,a,\nu}=u_{x_0,a,b,\nu}$ $\leb^n$ \ae \ in 
   $Q_{\nu}(x_0,\varepsilon)= Q_{-\nu}(x_0,\varepsilon)$, one has that
   \begin{equation*} 
      \begin{split}
         W_2(x_0,b,a,-\nu)=W_2(x_0,a,b,\nu),
      \end{split}
   \end{equation*} 
   for every $x_0\in \Omega$, $a,b\in \R^d$ and $\nu\in \R^d$.  
\end{theorem} 

\begin{remark} 
   \label{rmk:new2-sbv-fonseca}
   Condition \eqref{eq:fonseca-condition-iv}, can be softened  to
   \begin{equation}
      \label{eq:fonseca-condition-iv-new}
      \begin{split}
         \frac{1}{C} &\bigg( \int_{A} |\nabla u|^p \dx  + 
         \int_{S(u)\cap A}\big( |u^+ -u^-| \big)\d\hausd^{n-1} 
         \bigg)\\& \leq \mathcal F(u,A)\\& \leq C\bigg( \int_{A} |\nabla u|^p \dx + 
         \int_{S(u)\cap A}\big( |u^+ -u^-| \big)\d\hausd^{n-1} 
         \bigg).
      \end{split}
   \end{equation}
   Indeed, let us suppose that $\Fcal $ satisfies only \eqref{eq:fonseca-condition-iv-new}. By the same theorem(Theorem~\ref{thm:integral_represenatation_sbv-sbv}) it is possible to represent $Fcal(u,A) + \hausd(S_{u}\cap A)$, thus by removing the subtracted part it is possible to represent $\Fcal$. 
\end{remark} 

Finally, let us recall also the following result:
\begin{theorem}[{\cite{MR1686747}}]
\label{thm:cortesani}
Assume  that $\partial  \Omega$  is locally  Lipschitz  and let  $u\in
\sbvp(\Omega,\R^m)$. for every $\eps >0$ there exists a function $v\in
\sbvp(\Omega,\R^n)$ such that
\begin{enumerate}
\item $S_v$ is essentially closed.  
\item $\overline {S_v}$ is a polyhedral set
\item  $\|u-v\|_{L^p}\leq \eps$
\item $\|\nabla u -\nabla v\|\leq \eps$
\item $|\hausd^{n-1}(S_u)- \hausd^{n-1}(S_v)|\leq \eps$
\item $v\in C^\infty(\Omega\setminus \overline {S_v})$
\end{enumerate}
\end{theorem}

\subsection{Hypothesis and Main Theorem} 

Given Theorem~\ref{amb_theorem_4.7-sbv}, it is natural to impose the following hypothesis.   

Let $g^{(1)}$ a monotone convex functions  such that there exists a constant $C$ such that
\begin{equation*} 
   \begin{split}
      g^{(1)}(t) \geq C \max(t^{p} -1,0)
   \end{split}
\end{equation*} 
and $g^{(2)}$ be a monotone concave function such that 
\begin{equation*} 
   \begin{split}
      g^{(2)}(t)\geq c > 0
      \qquad\text{and}\qquad \lim_{t\uparrow\infty } \frac{g^{(1)}(t)}{t} = +\infty .    
   \end{split}
\end{equation*} 
The typical example we have in mind is when $g^{(1)}(t):= t^{p}$ and  $g^{(2)}(t) :=1+ t^{\alpha }$, where $0<\alpha <1$ and $p>1$.  

Let $T_\varepsilon\uparrow \infty $ be such that $\varepsilon T_\varepsilon\downarrow 0$. We   denote
\begin{equation*}
   \begin{split}
      g_\varepsilon(x)=
      \begin{cases}
         g^{(1)}(\|x\|) &\text{if }\|x\|< T_\varepsilon,\\
         \frac{1}{\varepsilon} g^{(2)}(\varepsilon \|x\|) &\text{if }\|x\|\geq T_\varepsilon.
      \end{cases}
   \end{split}
\end{equation*} 
We will also assume that there exists a constant $C$ such that 
$g^{(1)}(T_\varepsilon)\leq \frac{C}{\varepsilon}g^{(2)}(T_\varepsilon \varepsilon)$, and that for every $M>0$ there exists a constant $C_{M}$ such that
\begin{equation*} 
   \begin{split}
      g_\varepsilon (M|t|)\leq C_{M} g_{\varepsilon }(|t|).
   \end{split}
\end{equation*} 

Let $( f_{\xi ,\varepsilon } )$ be a family of local interactions such that for every $\xi,\varepsilon$ it holds
\begin{equation} 
   \label{eq:hypothesis-alto-sbv}
   \begin{split}
      f_{\xi ,\varepsilon }(x,t) \lesssim C_{\xi} \big(g_{\varepsilon }(|t |)+1\big)
   \end{split}
\end{equation} 
with $\sum_{\xi\in \Z^d} |\xi|C_{\xi} < +\infty$, and such that for every $1\leq j\leq d$ it holds
\begin{equation} 
   \label{eq:hypothesis-stima-basso-primi-vicini-sbv}
   \begin{split}
      f_{e_{i},\varepsilon }(x,t) \gtrsim \big(g_{\varepsilon }(|t|) - 1\big) 
   \end{split}
\end{equation} 

We will assume also that there exists a constant $M<+\infty$ such that 
\begin{equation*} 
   \begin{split}
      \int_{\R} \exp\left(-g_{\varepsilon }(t)\right) \dt \leq M.
   \end{split}
\end{equation*} 

Let us now define the Hamiltonians as 
\begin{equation*}
   \begin{split}
      H(u,A,\varepsilon )= \sum_{{\xi\in
            \Z^N}}\sum_{x \in R^\xi_\varepsilon(A)}
      f_{ \xi,\varepsilon  }\Big(x,\frac{\varphi (x+\varepsilon \xi)-\varphi (x)}{ |\xi|}\Big)
   \end{split}
\end{equation*}
and 
\begin{equation*} 
   \begin{split}
      H_{\infty}({\varphi},A,\varepsilon):=\sum_{\xi\in\Z^{d}}\sum_{x\in A_{\varepsilon}} f_{\varepsilon ,\xi }(x,\nabla _{\xi }\varphi(x)).
   \end{split}
\end{equation*}

\begin{remark} 
   \label{rmk:new-sbv}
   Let $u\in \sbvp(\Omega)\cap L^{\infty}(\Omega)$. Then one can show  there exists an discretized $\varphi_{u,\varepsilon}$ such that 
   \begin{equation*} 
      \begin{split}
         \|u \|_{\sbvp} \lesssim \varepsilon^{d}\sum_{x\in \varepsilon \Z^{d}\cap \Omega} g_{\varepsilon}(\nabla \varphi_{u,\varepsilon}) \lesssim \|u \|_{\sbvp}.
      \end{split}
   \end{equation*} 
   Indeed,  whenever $u$ is piecewise in $C^{\infty}$, the statement is trivial.  In order to conclude the general case it is enough to use Theorem~\ref{thm:cortesani}. 
\end{remark} 

Let us discuss very informally the above hypothesis. 
The function $g_{\varepsilon }$ will play the role of $\|\cdot\|^{p}$ in Section~\ref{sec:sobolev-representation} and the conditions on $g^{(1)}$ and $g^{(2)}$ are in order to ensure the compactness and lower semicontinuity. 
Given that a discrete function can be interpolated by continuous functions, it does not make sense to talk about jump set. 
However, it makes sense to consider as a jump set, the set of points where the discrete gradient is bigger that a certain threshold $T_{\varepsilon}$. 
Indeed, if we were approximating a function with a jump, it is expected that the gradient would explode(in a neighbourhood of the jump set) like $\delta /\varepsilon $, where $\delta$ is the amplitude of the jump and $\varepsilon$ is the discretization parameter.  Thus $T_{\varepsilon}\uparrow \infty$. 
Indeed, suppose that the function we are approximating is $\delta \chi_{B}$, where $\delta $ is a small parameter and $B$ is the unit ball. Then the jump set would be the set of points where the gradient goes like $\frac{\delta}{\varepsilon}$. 
Thus in order to ``catch'' jumps of order $\delta$ one needs that the $\lim_{\varepsilon \downarrow 0 } T_{\varepsilon }\varepsilon \leq \delta$. 
Thus $\lim_{\varepsilon \downarrow 0} T_{\varepsilon}\varepsilon = 0$.

As in the previous section, one of the main steps will be to show that $F'_{\infty}=F'$ and that $F''_{\infty }=F''$. The basic intuition behind, is again a version of the interpolation lemma. 
As before, we will show that if one imposes ``closeness'' $v$ in $L^{p}(A)$ to some regular function $u$, then one can impose also the boundary condition by ``paying a very small price in energy''. 
More precisely, given a sequence $\{ v_{n}\}$ such that $v_{n}\to u$ in $L^{p}(A)$, where $A$ is an open set, then there exists a sequence $\{\tilde{v}_{n}\}$ such that $\tilde{v}_n \to u$,  such that $\tilde{v}_{n}|_{\partial \Omega} =u|_{\partial \Omega }$  and 
\begin{equation*} 
   \begin{split}
      \liminf_{n} \|\tilde{v}_{n}\|_{\sbvp(A)} \leq \liminf_{n} \|{v}_{n}\|_{\sbvp(A)}.
   \end{split}
\end{equation*}

\begin{rmk}
   \label{rmk:easy-conseq-potent-def-sbv}
   Let $f:[0,+\infty)\to [0,+\infty)$ be a monotone function. Then, it is immediate to have 
   \begin{equation*}
      f \big ( \frac{1}{N}\sum_{i=1}^N t_i)\leq \sum_i f(t_i),
   \end{equation*}
   where $t_{i} >0$. 
\end{rmk}

Similarly  as in Section~\ref{sec:sobolev-representation}, 
for every $ A \in \Acal (\Omega) $, we define the free-energy as
\begin{equation*} 
   \begin{split}
      F(u,A,\kappa,\varepsilon) := - \varepsilon ^{d} \log 
      \int_{\Vcal(u,A,\kappa)} \exp \Big(- 
      H(\varphi,A, \varepsilon) \Big)\d\varphi\\
      F_{\infty }(u,A,\kappa,\varepsilon) := - \varepsilon ^{d} \log 
      \int_{\Vcal_{\infty }(u,A,\kappa)} \exp \Big(- 
      H_\infty (\varphi,A, \varepsilon) \Big)\d\varphi
   \end{split}
\end{equation*} 
where 
\begin{equation*} 
   \begin{split}
      \Vcal(u,A,\kappa )&= \insieme{\varphi:A_{\varepsilon }\to  \R^m|\ \frac{\varepsilon ^{d}}{|A| 
            ^{d}}\sum_{x\in A_{\varepsilon}}|u-\varepsilon \varphi|^{p}\leq \kappa^{p} }\\
      \Vcal_{\infty }(u,A,\kappa )&= \insieme{\varphi:\varepsilon\Z^{d} \to  \R^m|\ \frac{\varepsilon ^{d}}{|A| 
            ^{d}}\sum_{x\in A_{\varepsilon}}|u-\varepsilon \varphi|^{p}\leq 
         \kappa^{p}, \text{ and } \varphi(x) = \varphi _{u,\varepsilon }(x)\ \forall x\not \in A_{\varepsilon} },
   \end{split}
\end{equation*} 
where $\varphi _{u,\varepsilon }$ is defined in \eqref{eq:def-varphi-X}.

Similarly as in Section~\ref{sec:sobolev-representation},  let us introduce the following notations:
\begin{equation*} 
   \begin{split}
      F' (u,A,\kappa)&:= \liminf_{\varepsilon\downarrow 0} F(u,A,\kappa,\varepsilon)\\
      F''(u,A,\kappa)&:= \limsup_{\varepsilon\downarrow 0} F(u,A,\kappa,\varepsilon)\\
      F' (u,A)&:= \lim_{\kappa\downarrow 0}\liminf_{\varepsilon\downarrow 0} 
      F(u,A,\kappa,\varepsilon)=\lim_{\kappa\downarrow 0}F' (u,A,\kappa)\\
      F''(u,A)&:= \lim_{\kappa\downarrow 0}\limsup_{\varepsilon\downarrow 0} 
      F(u,A,\kappa,\varepsilon)=\lim_{\kappa\downarrow 0}F''(u,A,\kappa) \\
      F_{\infty }' (u,A,\kappa)&:= \liminf_{\varepsilon\downarrow 0} F_{\infty }(u,A,\kappa,\varepsilon)\\
      F_{\infty }''(u,A,\kappa)&:= \limsup_{\varepsilon\downarrow 0} F_{\infty }(u,A,\kappa,\varepsilon)\\
      F_{\infty }' (u,A)&:= \lim_{\kappa\downarrow 0}\liminf_{\varepsilon\downarrow 0} 
      F_{\infty }(u,A,\kappa,\varepsilon)=\lim_{\kappa\downarrow 0}F_{\infty }' (u,A,\kappa)\\
      F_{\infty }''(u,A)&:= \lim_{\kappa\downarrow 0}\limsup_{\varepsilon\downarrow 0} 
      F_{\infty }(u,A,\kappa,\varepsilon)=\lim_{\kappa\downarrow 0}F_\infty ''(u,A,\kappa) \\
   \end{split}
\end{equation*}

We are now able to write the main result of this section. 
\begin{theorem} 
   \label{thm:main-representation-sbv-statmech-sbv}
   Assume the previous hypothesis and that $u\in \sbvp\cap L^{\infty}$. Then for every infinitesimal sequence
   $(\varepsilon_n)$ there exists a  subsequence $\varepsilon_{n_k}$ and  functions 
   $W_{1}:\Omega \times \R^{d\times m} \to \R $ and $W_{2}:\Omega \times \R^{m}\times S^{d-1}\to \R$ such that
   \begin{equation*} 
      \begin{split}
         F(u,A):= F'_{n_k}(u,A) =  F''_{n_k}(u,A) = \int_A W_{1}(x,\nabla u) \dx + \int_{S_u} W_{2}( x,u^{+}(x)-u^{-}(x),\nu_{u}(x)),
      \end{split}
   \end{equation*} 
   where the function $W_{1} $ is a quasiconvex function and $W_{2} $ is a BV-elliptic function and depend on the chosen subsequence $\{\varepsilon_{n_k}\}$.  
\end{theorem}

\subsection{Proofs} 
\label{sub:Proofs-sbv}

The next technical lemma is a version of Lemma~\ref{lemma:zig-zag-sobolev}, that asserts that finite difference quotients along any direction can be controlled by finite difference quotients along the coordinate directions.

\begin{lemma} 
   \label{lemma:zig-zag-sbv-sbv}
   Let $A\subset \mathcal A(\Omega) $ and set $A_\varepsilon= \insieme{x\in A: \dist (x,A)>2\sqrt N \varepsilon }$. 
   Then there exists a dimensional constant $C:=C(N)$ such that  for any $\xi\in \mathbb Z^N$ there holds 
   \begin{equation*}
      \begin{split}
         \sum_{x\in R^{e_i}_\varepsilon (A_\varepsilon)} g_\varepsilon \big({ \nabla }_\xi
         u(x ) \big) \leq C |\xi| \sum_{i=1}^N\sum_{x\in
            R^{e_i}_\varepsilon (A)} g_\varepsilon({{ \nabla }}_{e_i}u(x)).
      \end{split}
   \end{equation*}
\end{lemma} 

\begin{proof} 
   As in the proof of Lemma~\ref{lemma:zig-zag-sobolev}, let  $\xi \in \Z^{d}$.
   By decomposing it into coordinates, it is not difficult to notice that it can be written as 
   \begin{equation*} 
      \begin{split}
         \xi = \sum_{k=1}^{N_{\xi }} \alpha  _{k}(\xi )e_{i_{k}},
      \end{split}
   \end{equation*} 
   where $N_{\xi } \leq \delta |\xi|$ and $\alpha_k (\xi )\in \{ -1,1\}$. 
   Denote by 
   \begin{equation*} 
      \begin{split}
         \xi _{k} = \sum_{j=1}^{N_{\xi }} \alpha_{k}(\xi ),
      \end{split}
   \end{equation*} 
   hence $|\xi_{k}| \leq |\xi|$ for all $k$. Thus
   \begin{equation*} 
      \begin{split}
         \nabla _{\xi } u (x) = \frac{1}{|\xi  |}\sum_{k=1}^{N_{\xi }}  \nabla 
         _{\alpha _{k}(\xi )e_{i}}u(x +\varepsilon \xi_{k} )
      \end{split}
   \end{equation*} 
   Moreover, by the monotonicity of $g_{\varepsilon}$, we have
   \begin{equation*} 
      \begin{split}
         g_{\varepsilon }\Big( \frac{1}{N_{\xi }} \sum_{k=1}^{N_{\xi } } \nabla 
         _{\alpha _{k}(\xi )e_{i}}u(x +\varepsilon \xi_{k} )\Big)
         \leq 
         \sum_{k=1}^{N_{\xi }} g_{\varepsilon }\Big(\nabla 
         _{\alpha _{k}(\xi )e_{i}}u(x +\varepsilon \xi_{k} )\Big)
      \end{split}
   \end{equation*} 
   Finally by summing over all $\xi $, exchanging the sums and using the 
   equivalence of the norms \ie $|\xi| \leq N_{\xi }\leq d |\xi|$ one has the 
   desired result.

\end{proof}

\vspace{3mm}

As in the previous section, let $G^{\lambda}$ be the free-energy (see \eqref{eq:def-free-energy} for the definition) induced by the Hamiltonian
\begin{equation*} 
   \begin{split}
      \tilde{H}^{\lambda}(\varphi,A,\varepsilon):=\lambda \sum_{i=1}^{d}\sum_{x\in R_{\varepsilon }^{e_{i}}(A)} g_{\varepsilon}(|\nabla_{i}\varphi |).
   \end{split}
\end{equation*} 

Similarly to Lemma~\ref{lemma:utile-stime-g}, one can prove

\begin{lemma} 
   \label{lemma:utile-stime-g-sbv}
   There exists  constants $C_{\lambda},D_{\lambda}$, such that it holds
   \begin{equation*} 
      \begin{split}
         C_{\lambda}|A| \leq G^{\lambda }(0,A,\kappa ,\varepsilon ) \leq D_{\lambda }|A|
      \end{split}
   \end{equation*} 
\end{lemma}

The next proof is the analog of Lemma~\ref{lemma:stima-basso}.
\begin{lemma} 
   \label{lemma:stima-basso-sbv}
   Let $\insieme{f_{\xi,\varepsilon }} $ satisfy the usual hypothesis.    Then there exists a constant $D >0 $ and $\varepsilon _{0}>0$ such that for every $\kappa <1$ it holds 
   \begin{equation} 
      \label{eq:stima-basso-1-sbv-sbv}
      \begin{split}
         \exp\left(-\varepsilon ^{-d} F(u,A,\kappa ,\varepsilon )\right) \leq  
         \exp\left( D|A| \varepsilon ^{-d} + D 
            \sum_{\xi \in R^{e_{i}}_{\xi }(A)}\sum_{i=1}^{d} g_{\varepsilon 
            }(\nabla _{e_{i}} \varphi _{u,\varepsilon }) \right),
      \end{split}
   \end{equation} 
   where $\varphi_{u,\varepsilon }$ is defined in \eqref{eq:def-varphi-X}.
\end{lemma} 

\begin{proof} 
   Given that  $g_{\varepsilon }(|a|)\lesssim g_{\varepsilon }(|a-b |) + g_{\varepsilon }(|b|)$ 
   one has that there exist constants $C_1$  such that 
   \begin{equation*} 
      \begin{split}
         H(\varphi ,A,\varepsilon ) & \geq C_{1}\sum_{i=1}^{d}\sum_{x\in R^{e_{i} 
            }_{\varepsilon }{(A)}} g_{\varepsilon }( |\nabla_{e_i}\varphi(x)  | ) \\ & \geq 
         C_{1} \sum_{i=1}^{d}  \sum_{x\in R^{e_{i}  }_{\varepsilon }{(A)}} 
         g_{\varepsilon }(|\nabla \psi |) -\tilde{C}_{1} \sum_{i=1}^{d}\sum_{x\in R^{e_{i}}_{\varepsilon }{(A)}} g_{\varepsilon }( |\nabla_{e_i}\varphi _{u,\varepsilon }(x) | )  
      \end{split}
   \end{equation*} 
   where $\psi = \varphi - \varphi _{ u,\varepsilon }  $.  Hence the estimate 
   \eqref{eq:stima-basso-1-sbv-sbv} reduces to prove that there exists a constant $D$  
   such that 
   \begin{equation*} 
      \begin{split}
         \int_{\insieme{\|\varepsilon\varphi \| \leq \kappa }} \exp\left(- C 
            \sum_{i=1}^{d}\sum_{x\in R^{e_{i} 
               }_{\varepsilon }{(A)}} g_{\varepsilon }( |\nabla_{e_i}\varphi | )\right) \leq 
         \exp\left( D|A| \varepsilon ^{-d}\right).
      \end{split}
   \end{equation*} 

   The above follows from Lemma~\ref{lemma:utile-stime-g-sbv}.
\end{proof}

As in Remark~\ref{rmk:almost-monotonicity}, we have the following: 
\begin{remark} 
   \label{rmk:almost-monotonicity-sbv}
   Let $u\in L^{\infty}\cap \sbvp$, then along the lines of Lemma~\ref{lemma:stima-basso-sbv}  one can easily prove that 
   there exists a constant $C$ such that 
   \begin{equation} 
      \label{eq:monotonicity_hypthesis-sbv}
      \begin{split}
         A\mapsto F'(u,A) + C(\|u\|_{\sbvp(A)} + |A | ) \qquad  A\mapsto F''(u,A)+ C(\|u\|_{\sbvp(A)} + |A |)
      \end{split}
   \end{equation} 
   are monotone with respect to the inclusion relation.  

   Without loss of generality we may assume that $F'$ and $F''$ satisfy \eqref{eq:monotonicity_hypthesis-sbv}. Indeed, recall that there is a representation formula as in the claim of Theorem~\ref{thm:main-representation-sobolev-statmech} for 
   $$F_{\{ \varepsilon_{n_k}\}}(u,A) + C(\|u\|_{\sbvp(A)} + |A | ),$$
   if and only if there is a representation formula for $F_{\{ \varepsilon_{n_k}\}}(u,A)$. 
\end{remark}

\begin{lemma} 
   \label{lemma:stima-alto-sbv}
   Let $f_{\xi ,\varepsilon }$ satisfy our hypothesis and let $A $ be an open 
   set. Then there exists  a constant $C,D>0$  such that 
   \begin{equation*} 
      \begin{split}
         \exp\left(-\varepsilon ^{-d} F(u,A,\kappa ,\varepsilon )\right) \geq  \exp\left( - D|A| \varepsilon ^{-d} -  C\sum_{i=1}^{d}\sum_{x\in R^{e_{i}}_{\varepsilon }} g_{\varepsilon }\big( |\nabla 
            _{e_i}\varphi_{u,\varepsilon } (x) | \big)  \right)     
      \end{split}
   \end{equation*} 
   where $\varphi_{u,\varepsilon }$ is defined in \eqref{eq:def-varphi-X}.
\end{lemma} 

\begin{proof} 
   Using Lemma~\ref{lemma:zig-zag-sbv-sbv}, one has that there exists a 
   constant $C$ such that 
   \begin{equation*} 
      \begin{split}
         H(\varphi ,A,\varepsilon ) \leq C \sum_{i=1}^{d}\sum_{x\in R^{e_{i} }_{\varepsilon }{(A)}} g_{\varepsilon }(|\nabla _{e_{i}}\varphi |)
      \end{split}
   \end{equation*}

   Given that 
   $g_{\varepsilon }(a+b)\leq g_{\varepsilon }(2a)+ g_{\varepsilon }(2b)\lesssim g_{\varepsilon }(a)+g_{\varepsilon }(b)$, there exist a constant $C_{1} $  such that
   \begin{equation*} 
      \begin{split}
         H(\varphi,A,\varepsilon)\leq 
         C_{1} \sum_{i=1}^{d}\sum_{x\in R^{e_i}_{\varepsilon }(A)} (  g_{\varepsilon }( |\nabla_{e_i} \varphi _{\varepsilon ,u}| )  +1 ) +2d\sum_{i=1}^{d}\sum_{x\in A_{\varepsilon }}  g_{\varepsilon }( |\nabla_{e_i}\psi(x) | ),
      \end{split}
   \end{equation*} 
   where $\psi = \varphi - \varphi _{u,\varepsilon }  $. 
   Hence, the estimate \eqref{eq:stima-basso-1-sbv-sbv} reduces to prove that there exists a constant $D $ such that 
   \begin{equation*} 
      \begin{split}
         \int_{\Vcal(0,A,\kappa ,\varepsilon )} \exp\left(- C
           \sum_{i=1}^{d} \sum_{x\in A_{\varepsilon}} g_{\varepsilon }( |\nabla _{e_i}\psi(x) | )\right) \geq ( \varepsilon 
         \kappa ) ^{-d}  
         \exp\left( D|A| \varepsilon ^{-d}\right).
      \end{split}
   \end{equation*} 

   The above inequality was proved in Lemma~\ref{lemma:utile-stime-g}.

\end{proof}

\begin{lemma}[exponential tightness] 
   \label{lemma:exponential-tightness-sbv}
   Let $A$ be an open set and $K\geq 0$.   Denote by 
   \begin{equation*} 
      \begin{split}
         \Mcal_{K} := \insieme{\varphi:\ H(\varphi,A,\varepsilon) \geq K 
            \varepsilon^{-d} |A|}.
      \end{split}
   \end{equation*} 
   Then there exists a constant $ D,K_{0},\varepsilon_{0} $ such that for every 
   $ K\geq K_{0},\ \varepsilon\leq \varepsilon_0  $   it holds
   \begin{equation*} 
      \begin{split}
         \int _{\Mcal _{K}\cap \Vcal(u,A,\kappa)} \exp \big( -H(\varphi ,A,\varepsilon )  \big) \leq \exp\Big(- \frac{1}{2}K\varepsilon^{-d}+ D{\varepsilon^{-d}}-D\sum_{i-1}^{d}\sum_{x\in R_{\varepsilon }^{e_i}(A)}g_{\varepsilon}(|\nabla _{e_i}\varphi_u |)\Big)
      \end{split}
   \end{equation*} 
\end{lemma}

\begin{proof} 
   For every $\varphi \in \Mcal_{K} $ it holds
   \begin{equation*} 
      \begin{split}
         H(\varphi,A,\varepsilon) \geq K/2 \varepsilon^{-d} + H(\varphi,A,\varepsilon).
      \end{split}
   \end{equation*} 
   Hence, by using Lemma~\ref{lemma:stima-alto-sbv}  we have the desired result.

\end{proof}

The proof of the following lemma is similar to Lemma~\ref{lemma:regularity}.

\begin{lemma}[regularity] 
   \label{lemma:regularity-sbv}
   Let $f_{\xi} $ satisfy the usual hypothesis and $u\in \sbvp\cap L^{\infty}$ then
   \begin{equation*} 
      \begin{split}
         \sup_{A'\Subset A} F''(u,A') = F''(u,A).
      \end{split}
   \end{equation*} 
\end{lemma}

\begin{proof} 

   Let us fix $A'\Subset A$ and $N\in \N$ ({to be chosen later}).  Let $\delta =\dist(A',A^C)$, and let 
   $0< t_{1},\ldots,t_N\leq \delta$ such that $t_{ i+1 }-t_{i}> \frac{\delta }{2N}$. 
   Without loss of generality, we may assume that there exists no 
   $x\in A_{\varepsilon }$ such that $\dist(x,A^{C})=t_i$. 
   For every $i
   $ we define
   \begin{equation*} 
      \begin{split}
         A_{i}:=\insieme{x\in A_{\varepsilon }:\ \dist(x,A^{C})\geq t_i}
      \end{split}
   \end{equation*} 
   and
   \begin{equation*} 
      \begin{split}
         S^{\xi ,\varepsilon }_{i}:=\insieme{x\in  ( A_{i} )_{\varepsilon }:\ x+\varepsilon \xi 
            \in A\setminus A_{i} }.
      \end{split}
   \end{equation*}

   We have that
   \begin{equation*} 
      \begin{split}
         R^{\xi }_{\varepsilon }(A) =  R^{\xi }_{\varepsilon }(A') +R^{\xi 
         }_{\varepsilon }(A\setminus \bar{A'}) + S^{\varepsilon ,\xi }_{i}
      \end{split}
   \end{equation*} 
   thus
   \begin{equation*} 
      \begin{split}
         H^{\xi }(\varphi ,A,\varepsilon ) = H^{\xi } (\varphi ,A\setminus 
         \bar{A_{i}},\varepsilon ) + H^\xi (\varphi ,A_{i},\varepsilon ) + 
         \sum_{x\in  S^{\xi ,\varepsilon }_{i}} f_{\xi ,\varepsilon }(\nabla _{\xi }\varphi (x)).
      \end{split}
   \end{equation*} 

   Hence,
   \begin{equation*} 
      \begin{split}
         H(\varphi ,A ,\varepsilon ) &= H(\varphi ,A_{i},\varepsilon ) + 
         H(\varphi ,A\setminus A_{i},\varepsilon ) + \sum_{\xi \in \Z^d}\sum_{x\in  S_{i}^{\xi 
               ,\varepsilon }} C_{\xi }\Big(g_{\varepsilon }\big(|\nabla _{\xi }\varphi (x) |\big)+1\Big)
      \end{split}
   \end{equation*} 

   Let us now estimate the last term in the previous inequality.  

   We separate the sum into two terms
   \begin{equation} 
      \label{eq:10121381574063-sbv}
      \begin{split}
         \sum_{\xi \in \Z^{d}} \sum_{x\in S^{\xi,\varepsilon  }_{i}} g_{\varepsilon }( |\nabla _{\xi }\varphi (x)| )= &
         \sum_{|\xi|\leq M} \sum_{x\in S^{\xi,\varepsilon  }_{i}} g_{\varepsilon }( |\nabla _{\xi }\varphi (x)| )+
         \sum_{ |\xi  |> M} \sum_{x\in S^{\xi,\varepsilon  }_{i}} g_{\varepsilon }( |\nabla _{\xi }\varphi (x)| ).
      \end{split}
   \end{equation} 

   Let $M\in \N$. 
   From the condition~\eqref{eq:hypothesis-alto-sbv} and  by taking $M$ sufficiently large,  we may also assume  without loss of  generality that 
   \begin{equation*} 
      \begin{split}
         \sum_{|\xi|\geq M} |\xi|C_{\xi} \leq \delta_1.
      \end{split}
   \end{equation*} 
   Hence, by using Lemma~\ref{lemma:zig-zag-sbv-sbv} we have that
   \begin{equation*} 
      \begin{split}
         \sum_{|\xi|\geq M} \sum_{x\in S^{\xi ,\varepsilon }_{i}} g_{\varepsilon }( |\nabla _{\xi }\varphi (x) | )\leq
         C \delta_1   \sum_{k=1}^{d} \sum_{x\in R^{e_{k} }_{\varepsilon }(A)} 
         g_{\varepsilon }( |\nabla _{e_{k} }\varphi (x)  | ) \leq \tilde{C}\delta _{1} H(\varphi ,A,\varepsilon ),
      \end{split}
   \end{equation*} 
   where in the last inequality we have used hypothesis~\eqref{eq:hypothesis-stima-basso-primi-vicini-sbv}.

   Let $|\xi|< M$.  If  $\varepsilon MN\leq 2\delta $, then for every 
   \begin{equation*} 
      \begin{split}
         S^{\xi,\varepsilon}_{i} \cap S^{\xi,\varepsilon}_{j}=\emptyset \qquad 
         \text{ whenever  } |i-j |\geq 2.
      \end{split}
   \end{equation*} 
   As the we are interested in $\varepsilon\to 0$, we may assume the above without loss of generality. 

   Given that
   \begin{equation*} 
      \begin{split}
         \sum_{i=1}^{N-2} \sum_{|\xi| < M}\sum_{x\in  S_{i}^{\varepsilon ,\xi }} 
         g_{\varepsilon }( |\nabla \varphi (x) | )\leq  2C H(\varphi ,A,\varepsilon )
      \end{split}
   \end{equation*} 
   there exist $0<i\leq N-2$ such that 
   \begin{equation} 
      \label{eq:cagata1-sbv}
      \begin{split}
         \sum_{|\xi| <M} \sum_{x\in S_{i}^{\xi ,\varepsilon }}g_{\varepsilon }( |\nabla_{\xi }\varphi | )<\frac{2C}{N-2}H(\varphi ,A,\varepsilon ).
      \end{split}
   \end{equation} 

   Let us denote by $\Ncal_{i}$ the set of all 
   $\varphi\in \Vcal(u,A,\kappa,\varepsilon )$ such that 
   \eqref{eq:cagata1-sbv} holds 
   for the first time, namely for every $j\leq i$ 
   \begin{equation} 
      \label{eq:cagata2-sbv}
      \begin{split}
         \sum_{|\xi| <M} \sum_{x\in S_{i}^{\xi ,\varepsilon }}g_{\varepsilon }( |\nabla_{\xi }\varphi | ) <  \frac{2C}{N-2}H(\varphi ,A,\varepsilon ).
      \end{split}
   \end{equation} 
   On one side, we have that
   \begin{equation*} 
      \begin{split}
         \int_{\Vcal(u,A,\kappa ,\varepsilon )} \exp\left( -H(\varphi ,A,\varepsilon )   \right) \leq \sum_{i=1}^{N}\int_{\Ncal_{i}}\exp\left( - H(\varphi ,A_{i},\varepsilon) - H(\varphi ,A\setminus \bar{A}_{i},\varepsilon )\right),
      \end{split}
   \end{equation*} 
   on the other side one has that
   \begin{equation*} 
      \begin{split}
         \int _{\Vcal(u,A,\kappa ,\varepsilon )} \exp\left( -H(\varphi ,A,\varepsilon )   \right) \geq 
         \sum_{i=1}^{N}\int _{\Ncal^{K}_{i}}\exp\left( -H(\varphi ,A,\varepsilon )  \right) ,
      \end{split}
   \end{equation*} 
   where $\Ncal^{K}_{i}:= \Ncal_{i}\setminus \Mcal_{K}$.  
   By using  \eqref{eq:cagata2-sbv}, one has that for every $\varphi \in  \Ncal_{i}^{K}$ it holds
   \begin{equation*} 
      \begin{split}
        H(\varphi ,A,\varepsilon )\leq H(\varphi, A_{i}) + H(\varphi,A\setminus \bar{A}_{i}) + \frac{K|A|\varepsilon^{-d}}{N-2} 
      \end{split}
   \end{equation*} 
   and for every $\varphi$ it holds
   \begin{equation*} 
      \begin{split}
         H(\varphi ,A,\varepsilon )\geq H(\varphi, A,\varepsilon ) + H(\varphi ,A\setminus \bar{A}_i,\varepsilon ) .
      \end{split}
   \end{equation*} 

   Hence,
   \begin{equation*} 
      \begin{split}
         \int _{\Vcal(u,A,\kappa ,\varepsilon )} \exp\left( -H(\varphi ,A,\varepsilon )     \right) \geq  
         \sum_{i=1}^{N}\int _{\Ncal^{K}_{i}}\exp\left( - H(\varphi, A_{i}) - H(\varphi 
            ,A\setminus \bar{A}_{i}) - \frac{K|A|\varepsilon^{-d}}{N-2} \right).
      \end{split}
   \end{equation*}

   From now on the proof follows as in Lemma~\ref{lemma:regularity}.
\end{proof}

\begin{lemma} 
   \label{lemma:cambio-dati-bordo-sbv}
   For every open set $A$  and $u\in \sbvp(\R^{d})\cap L^{\infty}$  it holds
   \begin{equation*} 
      \begin{split}
         F'(u,A)=F'_{\infty }(u,A)\qquad\text{and}\qquad  F'(u,A)=F'_{\infty }(u,A).
      \end{split}
   \end{equation*} 
\end{lemma}

\begin{proof} 

	The proof of the above statement follows in the same way as in Lemma~\ref{lemma:cambio-dati-bordo}.

\end{proof}

\begin{lemma}[subadditivity] 
   \label{lemma:subadditivity-sbv}
   Let $A',A,B',B \subset \Omega $ be open sets such that $A'\Subset A $ and 
   such that $B' \Subset B $.  Then for every $u\in \sbvp\cap L^{\infty}$ one has that
   \begin{equation*} 
      \begin{split}
         F''(u,A'\cup B') \leq F''(u,A) + F''(u,B)
      \end{split}
   \end{equation*} 
\end{lemma} 

\begin{proof} 
   The proof of this statement is very similar to 
   Lemma~\ref{lemma:regularity-sbv} and Lemma~\ref{lemma:cambio-dati-bordo-sbv} .
\end{proof}

\begin{lemma}[locality] 
   \label{lemma:locality-sbv}
   Let $u,v \in \sbvp(\Omega)\cap L^{\infty}$ such that $u\equiv v $ in $A$.   Then 
   \begin{equation*} 
      \begin{split}
         F'(u,A)=F'(v,A)\qquad \text{and}\qquad F''(u,A)=F''(v,A)
      \end{split}
   \end{equation*} 
\end{lemma} 

\begin{proof} 
   The statement follows from the definitions.  

\end{proof} 

\begin{proof}[Proof of {Theorem~\ref{thm:main-representation-sbv-statmech-sbv}}]\    
   Let us suppose initially that there exists a sequence for which 
   $F(\cdot,\cdot)=F'(\cdot,\cdot)= F''(\cdot,\cdot)$.  Then to conclude it is enough to notice that $F$ satisfies the conditions of Theorem~\ref{thm:integral_represenatation_sbv-sbv}, which are proved in the previous Lemmas.  
\end{proof}

\begin{corollary} 
   Because of Lemma~\ref{lemma:cambio-dati-bordo-sbv}, the same statement holds true for $F_{\infty}$. 
   This in particular implies that for the sequence $\insieme{\varepsilon _{n_{k}}}$ in Theorem~\ref{thm:main-representation-sbv-statmech-sbv}  there holds a large deviation principle with rate functional 
   \begin{equation}
      I(v)=\int_{\Omega} W_{1}(x,\nabla v) \dx \int_{Ju}W_{2}(x,u_{+}(x)-u_{-}(x)) \d \hausd^{d-1}(x) - \min_{\bar v\in W^{1,p}_0(\Omega)+u}\,  
      \int_{\Omega} W(\nabla \bar v(x))dx. 
   \end{equation} 
\end{corollary}

\end{document}